\documentclass[aps,prl,twocolumn,superscriptaddress,eprint]{revtex4-2}
\usepackage{graphics,graphicx,array,afterpage}
\usepackage{qcircuit,amsmath}
\usepackage{grffile}
\usepackage[export]{adjustbox}
\usepackage{lineno}
\usepackage{xcolor}
\usepackage{longtable}
\usepackage{braket}
\usepackage{multirow}
\usepackage{enumerate}
\usepackage[colorlinks, linkcolor=blue, citecolor=blue]{hyperref}
\usepackage{amsmath}
\usepackage{amssymb}
\usepackage{amsthm,times}
\usepackage{pifont}

\newtheorem{theorem}{Theorem}

\newcommand{\tr}{\mbox{tr}}

\begin{document}

\title{ Expressibility-induced Concentration of  Quantum Neural Tangent Kernels}

\author{Li-Wei Yu}
\email{yulw@nankai.edu.cn}
\affiliation{Theoretical Physics Division, Chern Institute of Mathematics and LPMC, Nankai University, Tianjin 300071, P. R. China}
\author{Weikang Li}
\author{Qi Ye}
\author{Zhide Lu}
\author{Zizhao Han}
\affiliation{Center for Quantum Information, IIIS, Tsinghua University, Beijing 100084, P. R. China}
\author{Dong-Ling Deng}
\email{dldeng@tsinghua.edu.cn}
\affiliation{Center for Quantum Information, IIIS, Tsinghua University, Beijing 100084, P. R. China}
\affiliation{Hefei National Laboratory, Hefei 230088, People’s Republic of China}
\affiliation{Shanghai Qi Zhi Institute, 41st Floor, AI Tower, No. 701 Yunjin Road, Xuhui District, Shanghai 200232, China}
\begin{abstract}
Quantum tangent kernel methods provide an efficient approach to analyzing the performance of quantum machine learning models in the infinite-width limit, which is of crucial importance in designing appropriate circuit architectures for certain learning tasks. Recently,  they have been adapted to describe the convergence rate of training errors in quantum neural networks in an analytical manner. Here, we study the connections between the trainability and expressibility of quantum tangent kernel models. In particular,  for global loss functions,  we rigorously prove that  high expressibility of both the global and local quantum encodings can lead to exponential concentration of quantum tangent kernel values to zero. Whereas for local loss functions,  such issue of exponential concentration persists owing to the high expressibility, but can be partially mitigated. We further carry out extensive numerical simulations to support our analytical theories. Our discoveries unveil a pivotal characteristic of quantum neural tangent kernels, offering valuable insights for the design of wide quantum variational circuit models in practical applications.
\end{abstract}

\maketitle

Quantum-enhanced machine learning is a rapidly evolving field that seeks to combine the  power of quantum computation with classical machine learning algorithms \cite{Wittek2014Quantum, Schuld2021Machine}. This convergence promises a notable quantum advantage, empowering machine learning methodologies to address the challenges of processing big data in real-world scenarios, thereby surpassing their classical counterparts. 
Along this line, a number of pioneering works have been conducted, ranging from theoretical frameworks to practical algorithmic developments  \cite{Biamonte2017Quantum, Harrow2009Quantum, Lloyd2014Quantum, Dunjko2016Quantum, Amin2018Quantum,  Havlicek2019Supervised, Schuld2019Quantum, Grant2018Hierarchical, Cong2019Quantum, Li2021Recent,Gao2018Quantum, Lloyd2018Quantum,Lu2020Quantum,Liu2020Vulnerability,Liu2021Rigorous, Huang2021Quantum,Gao2022Enhancing,Anschuetz2023Interpretable,Mcclean2018Barren, Wang2021Noise, Cerezo2021Higher, Sharma2022Trainability, Arrasmith2021Effect, Holmes2021Barren, Marrero2021Entanglement, Patti2021Entanglement, Pesah2021Absence, Arrasmith2022Equivalence, Cerezo2021Cost, Uvarov2021Barren, Grant2019Initialization, Zhao2021Analyzing,Banchi2021Generalization, Caro2022Generalization, Caro2023Out, Du2023Problem,Poland2020No, Du2021Learnability, Sharma2022Reformulation,Liu2023Analytic, Rad2023Deep, Garcia-Martin2023Deep,Wu2021Expressivity,Lu2021Markovian,Marshall2023High,Manzano2023Parametrized,Li2022Quantum,Gyurik2023Exponential,Gyurik2023Exponential,Jerbi2023Quantum,Ye2023Quantum,Larasati2022Quantum,Shingu2022Variational,Yun2022Slimmable,Hirche2023Quantum,Gan2022Fock,Yang2023Analog,Li2022Concentration,Chen2022Variational,Bittel2021Training,Bittel2022Optimizing}.
Within the realm of quantum algorithms, significant progress has been made, encompassing quantum principal component analysis \cite{Lloyd2014Quantum}, quantum neural networks \cite{Grant2018Hierarchical, Cong2019Quantum, Li2021Recent}, quantum generative models \cite{Lloyd2018Quantum, Gao2018Quantum}, quantum adversarial learning \cite{Lu2020Quantum,Liu2020Vulnerability}, {\it etc.} Several of these innovations have been experimentally validated, showcasing the remarkable potential of quantum learning techniques for practical problem-solving and highlighting the distinctive attributes inherent to quantum models, thereby opening up exciting avenues for quantum-based approaches \cite{Hu2019Quantum, Ren2022Quantum, Pan2023Deep}.
In the field of quantum learning theory, a series of notable results have been achieved, encompassing interpretable quantum advantages \cite{Liu2021Rigorous, Huang2021Quantum,Gao2022Enhancing,Anschuetz2023Interpretable}, barren plateaus \cite{Mcclean2018Barren, Wang2021Noise, Cerezo2021Higher, Sharma2022Trainability, Arrasmith2021Effect, Holmes2021Barren, Marrero2021Entanglement, Patti2021Entanglement, Pesah2021Absence, Arrasmith2022Equivalence, Cerezo2021Cost, Uvarov2021Barren, Grant2019Initialization, Zhao2021Analyzing}, generalization \cite{Banchi2021Generalization, Caro2022Generalization, Caro2023Out, Du2023Problem}, learnability \cite{Poland2020No, Du2021Learnability, Sharma2022Reformulation}, and training dynamics \cite{Liu2023Analytic, Rad2023Deep, Garcia-Martin2023Deep}, {\it etc.}
Moreover, this exploration has introduced concepts from quantum physics, including entanglement, nonlocality, and quantum contextuality \cite{Anschuetz2023Interpretable, Gao2022Enhancing}, to elucidate the rationale behind the quantum advantage in learning models. These results deepen our understanding on the theory of quantum machine learning. For instance, the volume law of entanglement entropy in a quantum neural network can give rise to the challenging barren plateau problem, thus impeding the efficient training of quantum machine learning models \cite{Marrero2021Entanglement}.
Here we introduce the concept of quantum expressibility, and study how quantum expressibility influences the trainability of quantum neural tangent kernel (QNTK) models.

\begin{figure}
\centering
\includegraphics[width=\linewidth]{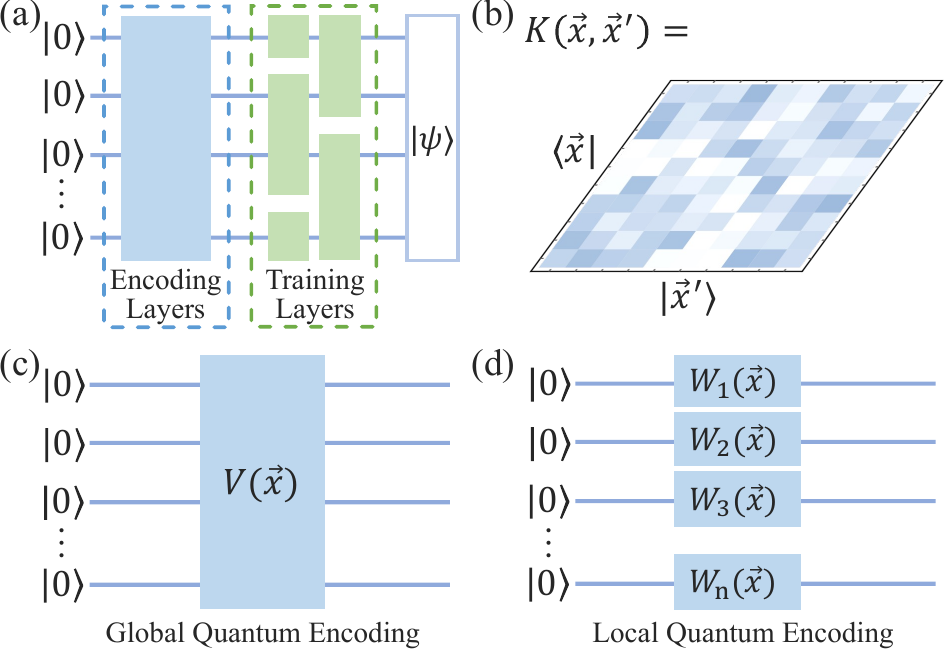} 
\caption{
{A schematic diagram of quantum machine learning circuit and the associated QNTK.
} 
(a) Quantum circuit of machine learning model. Input data samples are encoded through the encoding layers represented by the unitary $V(\vec{x})$, and the variational parameters for quantum machine learning are located in the training layers.
(b) Schematic description of the distribution of QNTK values in input data space.
(c) Global quantum encoding through the unitary $V(\vec{x})$.
(d) Local quantum encoding through the unitary  $W(\vec{x})=\otimes_{i=1}^n W_i$.
}
\label{Fig1_Blue}
\end{figure}

The introduction of the neural tangent kernel  provides a unique perspective  for examining the dynamical behavior of infinite-width deep neural networks and has gained considerable attention in recent years \cite{Jacot2018Neural, Lee2019Wide, Arora2019On, Du2019Graph, SohlDickstein2020On, Yang2020Neural, Yang2021Feature, Barzilai2023A}. It is crucial in uncovering the duality between the training dynamics of infinite-width deep neural networks and kernel methods,
thus providing predictability to learning capabilities. For instance, the eigenvalues of the neural tangent kernel can predict the generalization ability of machine learning models \cite{Simon2023The}. Moreover, inspired by the success in characterizing classical neural networks, there has been a surge of interest in adapting the QNTK approach to wide quantum neural networks \cite{Shirai2022Quantum, Liu2022Representation, Liu2023Analytic, Rad2023Deep, Garcia-Martin2023Deep, You2023Analyzing, Luo2023Infinite,Wang2023Symmetric}. 
 In particular, based on the theory of QNTK, one can analytically study the training dynamics of wide quantum neural networks \cite{Liu2022Representation}, thereby characterizing the convergence of residual training error \cite{Liu2023Analytic}. In addition, the QNTK can be employed for representation learning, with numerical results illustrating its superior performance compared to conventional quantum kernel methods when handling ansatz-generated datasets \cite{Shirai2022Quantum}. These intriguing findings highlight the indispensable role of QNTK in establishing the foundation of quantum machine learning theory. Yet, it is worth noting that the theory of QNTK is still in its infancy, with numerous critical issues remaining largely unexplored.


In quantum machine learning, one pivotal step involves encoding input samples into quantum states using unitary quantum circuits, which concerns the encoding expressibility. 
The expressibility of a set of unitary elements is defined by how effectively the set uniformly spans the unitary group, determining the ability of unitary circuits to approximate arbitrary unitary operations \cite{Harrow2009Random, Sim2019Expressibility}. Recently, the expressibility has been shown connected to the barren plateau problem in quantum machine learning models \cite{Holmes2022Connecting, Thanasilp2022Exponential}. In this paper, by exploring the distribution of QNTK values in the input data space, we investigate the connections between the expressibility and trainability of QNTK models, which is an important but hitherto unexplored issue in the existing literature.   In particular, we consider the expressibility of both global and local quantum encodings, and rigorously prove that: (i) For global loss functions,  sufficient encoding expressibility will lead to an exponential concentration of QNTK values converging to zero, with respect to the system size. This concentration hinders the efficient estimation of QNTK values. (ii) For local loss functions, such expressibility-induced exponential concentration of QNTK values persists but can be partially mitigated by judiciously selecting appropriate observables. In addition, we provide extensive numerical results that validate the analytical predictions for various system sizes and circuit depths.

{\it Expressibility of quantum encoding.}--- To encode the $q$-dimensional classical data vector  ${\vec{x}}\in \mathbb{C}^q$ into the  quantum state $|{\vec{x}}\rangle = V({\vec{x}})|\psi_0\rangle$, where $|\psi_0\rangle$ denotes the computational basis, one needs a map from the vector space to the unitary group element $\mathbb{V}_{\vec{x}}: {\vec{x}}\in  \mathbb{C}^q \longmapsto V({\vec{x}})\in \mathcal{V}(q)$, with  $\mathcal{V}(q)$ a set of $q$-dimensional  unitaries.  Thus all of the feasible unitary elements $\mathbb{V}_{\vec{x}}$  constitute an ensemble $\mathbb{V}=\{\mathbb{V}_{\vec{x}}\}$. We focus on the global and local quantum encodings, as well as their expressibilities. 
The  expressibility of a global encoding ensemble $\mathbb{V}$ is defined by how close the ensemble $\mathbb{V}$ is to a unitary $t$-design  \cite{Sim2019Expressibility,Holmes2022Connecting,Thanasilp2022Exponential}
\begin{equation}\label{Expressibility}
\mathcal{A}^{[t]}_{\mathbb{V}} (\cdot)=\mathcal{D}^{[t]}_{\rm Haar}(\cdot)-\int_{\mathbb{V}}dVV^{\otimes t}(\cdot)^{\otimes t}(V^{\dag})^{\otimes t},
\end{equation}
where $\mathcal{D}^{[t]}_{\rm Haar}(\cdot) = \int_{\mathcal{U}(q)}d\mu(U) U^{\otimes t}(\cdot)^{\otimes t}(U^{\dag})^{\otimes t}$ represents the $t$-moment Haar integral of the unitary group $\mathcal{U}(q)$ \cite{Collins2006Integration}. One can define the data-dependent measure of the $t$-moment expressibility by the trace norm $\mathcal{M}^{[t]}_{\mathbb{V}} =\left \| \mathcal{A}^{[t]}_{\mathbb{V}}(\rho_0)\right \| _1$ \cite{Thanasilp2022Exponential}, 
given that the initial quantum state is $\rho_0$. When the ensemble $\mathbb{V}$ is $t$-moment maximally expressive, the term  $\mathcal{M}^{[t]}_{\mathbb{V}}=0$.  We note that if an ensemble forms the unitary $t$-design, it also forms the unitary $(t-1)$-design \cite{Renes2004Symmetric,Dankert2009Exact,Harrow2009Random}. Thus for the expressibility measure of an ensemble,  $\mathcal{M}^{[t]}_{\mathbb{V}}\rightarrow0$ indicates $\mathcal{M}^{[t-1]}_{\mathbb{V}}\rightarrow0$.  In our work, it suffices to consider the case of $t=2$, {\it i.e.}, the unitary 2-design.

Similarly, the expressibility of an $n$-qubit local encoding ensemble $\mathbb{W}$ is defined by how closely the ensemble $\mathbb{W}$ approximates a $t$-design of products of single-qubit unitaries:
\begin{equation}\label{Expressibility_Local}
\mathcal{A}^{[t]}_{\mathbb{W}} (\cdot) = \mathcal{D}^{[t]}_{{\rm Haar}_1^{\otimes n}}(\cdot) - \int_{\mathbb{W}}dWW^{\otimes t}(\cdot)^{\otimes t}(W^{\dag})^{\otimes t},
\end{equation}
where $\mathcal{D}^{[t]}_{{\rm Haar}_1^{\otimes n}}(\cdot)$ represents the $t$-moment integral over the products of Haar-random single-qubit unitaries. The measure of expressibility is $\mathcal{M}^{[t]}_{\mathbb{W}} = \left \| \mathcal{A}^{[t]}_{\mathbb{W}}(\rho_0)\right \| _1$.

 {\it Quantum neural tangent kernel.}--- In the context of quantum machine learning, the mean square loss function for a supervised learning task reads
\begin{equation}\label{MSE_Loss}
\mathcal{L} = \frac{1}{2}\sum_i \left(\langle{\vec{x}}_i |U^\dag(\vec{\theta})\hat{O}U(\vec{\theta})|{\vec{x}}_i\rangle - y_i\right)^2 = \frac{1}{2}\sum_i \varepsilon_i^2,
\end{equation}
where $\vec{\theta}$ represents the set of variational parameters,  ${\vec{x}}_i$ and $y_i$ represent the feature and the label of the $i$-th data sample, respectively. The training residual error, denoted by $\varepsilon_i = \langle{\vec{x}}_i |U^\dag(\vec{\theta})\hat{O}U(\vec{\theta})|{\vec{x}}_i\rangle - y_i$, serves as a fundamental quantity in this framework. The machine learning model is trained using the gradient descent algorithm, where the variational parameters are updated as $\delta \theta_{\lambda} = -\eta\frac{\partial \mathcal{L}(\vec{\theta})}{\partial \theta_{\lambda}}=-\eta \sum_i \varepsilon_i \frac{\partial \varepsilon_i}{\partial \theta_{\lambda}}$. When the learning rate $\eta$ is small, the residual error can be updated as $\delta\varepsilon_m \approx \sum_\lambda \frac{\partial \varepsilon_m}{\partial \theta_{\lambda}}\delta \theta_{\lambda} = -\eta \sum_{i,\lambda} \varepsilon_i\frac{\partial \varepsilon_m}{\partial \theta_{\lambda}} \frac{\partial \varepsilon_i}{\partial \theta_{\lambda}} = -\eta\sum_{i} K_{m,i}\varepsilon_i$, where the QNTK is introduced through \cite{Liu2022Representation}
\begin{equation}\label{QNTK}
K_{m,i}=K(\vec{x}_m, \vec{x}_i) =  \sum_{\lambda}\frac{\partial \varepsilon_m}{\partial \theta_{\lambda}} \frac{\partial \varepsilon_i}{\partial \theta_{\lambda}}.
\end{equation}
It is noteworthy that the corresponding kernel matrix in this setting is positive semi-definite. The QNTK plays a crucial role in describing the training dynamics of quantum machine learning models. Specifically, in the training process of wide quantum machine learning models, driven by the goal of minimizing the absolute values of the residual errors $\{\varepsilon_i\}$ using gradient descent on mean square loss functions, the QNTK characterizes the asymptotic training dynamics of the residual errors $\{\varepsilon_i\}$. In particular, in the regime of lazy learning \cite{Liu2022Representation}, where the variational parameters $\vec{\theta}$ undergo minimal changes, the training dynamics of the residual error $\varepsilon_i(t)$ at time $t$ can be approximated as $\varepsilon_i(t)\approx (1-\eta K)^t_{i,j}\varepsilon_j(0)$, with $\epsilon_j(0)$ representing the initial residual error at time $t=0$, and $\eta$ denoting the small learning rate. 
Therefore, by studying the properties of QNTK, one can conveniently predict the dynamical behaviors in the training process of the quantum learning models.

{\it Expressibility-induced concentration of QNTK.}---Akin to their classical counterparts, QNTK is crucial to predicting the training dynamics of wide quantum neural networks, which are defined as  variational quantum circuits with a multitude of qubits and variational parameters \cite{Liu2023Analytic}.  In a typical supervised quantum learning task involving classical data, the transformation of  the input set $\mathcal{X}$ into quantum states hinges on the expressibility of the quantum feature map. Here we delve into the connection between the expressibility of quantum encoding and the trainability of the quantum tangent kernel model, with a focus on the loss function defined in Eq.~(\ref{MSE_Loss}). We first consider the case of global encoding. 
 

 \begin{figure}
\centering
\includegraphics[width=\linewidth]{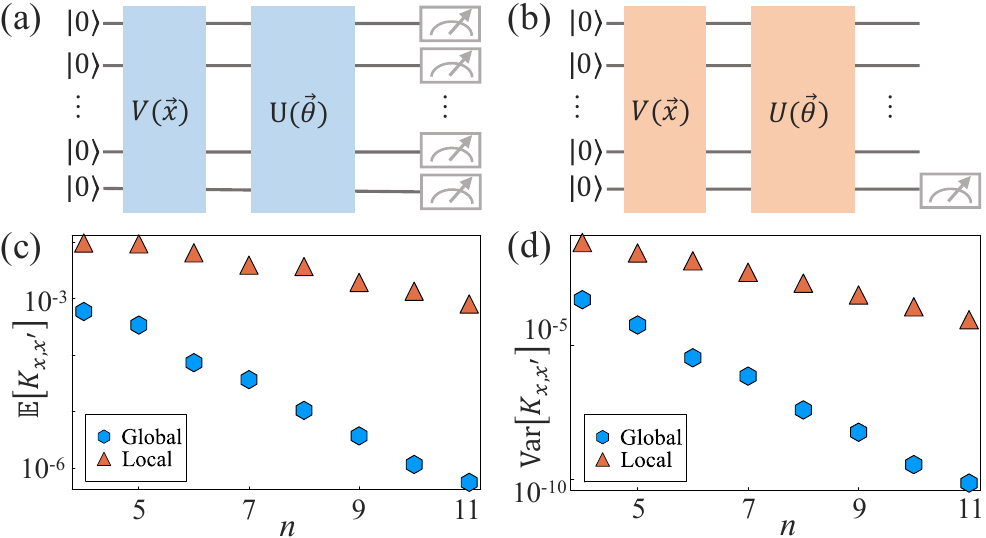} 
\caption{ {Numerical results on the mean and variance of QNTK values over the input data space.} 
Quantum learning models with (a) global observable  and (b) local observable, with the global encoding layers $V(\vec{x})$ and the variational layers $U(\vec{\theta})$. (c) Mean value and (d) Variance of QNTK versus the system size $n$. In both cases, the system size $n$ varies from four qubits to eleven qubits, and the number $d$ of variational layers  is proportional to the system size $n$ \cite{EICQNTK_supplementary}. 
}
\label{Fig2_var}
\end{figure}

\begin{theorem}[Global Encoding]\label{theorem:global} 
Given the quantum neural tangent kernel defined in Eq.~(\ref{QNTK}), for arbitrary two input samples $\vec{x}$ and $\vec{x}'$ drawn from the same distribution, the QNTK $K(\vec{x}, \vec{x}')$ adheres to the following inequality:
\begin{equation}\label{QNTK_Distribution}
\begin{aligned}
&\mathrm{Pr}_{\vec{x}, \vec{x}'}\left(|K(\vec{x}, \vec{x}')+\mathcal{O}(\mathcal{M}^{[1]}_{\mathbb{V}})|\geq \epsilon\right)\\ 
&\leq  \frac{2\epsilon^{-2}\Lambda^2}{(2^{2n}-1)^2}\left(\operatorname{tr}[ \hat{O}^2]\right)^2-\epsilon^{-2}\mathcal{O}((\mathcal{M}^{[1]}_{\mathbb{V}})^2,\mathcal{M}^{[2]}_{\mathbb{V}}),
\end{aligned}
\end{equation}
for arbitrary small positive constant $\epsilon$, where $n$ is the number of qubits, $\Lambda$ is the number of trainable parameters, and $\mathcal{M}^{[1,2]}_{\mathbb{V}}$ represents the measure of the $\{1,2\}$-moment expressibility of the global encoding ensemble $\mathbb{V}$.
\end{theorem}

\begin{proof}
Here we present a concise overview of the main idea. Due to its technical intricacy, a comprehensive proof is provided in the Supplementary Materials \cite{EICQNTK_supplementary}. Our result  is based on the Chebyshev's inequality. To begin with, we compute the mean value of the QNTK across all feasible input samples. Utilizing the expressibility defined in Eq.~(\ref{Expressibility}), we express the mean value of $K(\vec{x}, \vec{x}')$ in terms of the 1-moment expressibility $\mathcal{M}^{[1]}_{\mathbb{V}}$, which tends towards zero for a highly expressive ensemble $\mathbb{V}$. Subsequently, we proceed to calculate the variance of QNTK using the 2-moment Haar integral of the unitary group. Direct calculations reveal that the variance of $K(\vec{x}, \vec{x}')$ can be expressed in terms of the 1-moment and 2-moment expressibilities. This leads to the result in Ineq.~(\ref{QNTK_Distribution}) and  completes the proof.
\end{proof}

In comparison to the global encoding, local quantum encoding typically demands lower circuit depth and exhibits greater scalability with increasing system size. This can be particularly advantageous when dealing with limited quantum hardware resources. We now explore how the expressibility of local encoding impacts the concentration of QNTK values.

\begin{theorem}[Local Encoding]\label{theorem:local}
For local quantum encoding, in the case of arbitrary two input samples $\vec{x}$ and $\vec{x}'$ drawn from the same distribution, the QNTK $K(\vec{x}, \vec{x}')$ obeys
\begin{equation}\label{QNTK_Distribution_Local}
\begin{aligned}
&\mathrm{Pr}_{\vec{x}, \vec{x}'}\left(|K(\vec{x}, \vec{x}')+\mathcal{O}(\mathcal{M}^{[1]}_{\mathbb{W}})|\geq \epsilon\right)\\
&\leq  \frac{\epsilon^{-2}\Lambda^2}{2^{2n-2}}\left(\operatorname{tr}[ \hat{O}^2]\right)^2-\epsilon^{-2}\mathcal{O}((\mathcal{M}^{[1]}_{\mathbb{W}})^2,\mathcal{M}^{[2]}_{\mathbb{W}}).
\end{aligned}
\end{equation}
Here, $\epsilon$ represents arbitrary small positive constant, $n$ denotes the number of qubits, $\Lambda$ denotes the number of trainable parameters, and $\mathcal{M}^{[1,2]}_{\mathbb{W}}$ denotes the measure of the $\{1,2\}$-moment expressibility of the local encoding ensemble $\mathbb{W}$.
\end{theorem}

\begin{proof}
We outline the key steps here and leave the technical details to the Supplementary Materials \cite{EICQNTK_supplementary}. First, we compute the mean value of QNTK over all potential input samples, expressing them in terms of the 1-moment expressibility of local encoding as defined in Eq.~(\ref{Expressibility_Local}). Subsequently, we calculate the QNTK variance and express it using the 1-moment and 2-moment expressibilities of local encoding. By leveraging Chebyshev's inequality, we establish the connections between the distribution of QNTK values in the input data space and the expressibility of local encoding.
\end{proof}

The aforementioned theorems establish straightforward connections between the expressibility and the distribution of QNTK values in input data space, covering global and local encoding schemes. It is important to note that the terms $\mathcal{M}^{[1,2]}_{\mathbb{V}}$ ($\mathcal{M}^{[1,2]}_{\mathbb{W}}$) approach zero for highly expressive quantum encoding. In such a vein, we consider two typical sorts of loss functions: the global and local ones. For global loss functions, exemplified by the use of the observable $\hat{O} = (|0\rangle\langle0|)^{\otimes n}$, the probability that the QNTK value deviates from zero within a given precision $\epsilon$ decreases exponentially with respect to the increasing number of physical qubits. This exponential concentration renders the QNTK method inefficient and impractical for describing the training dynamics of infinite-width quantum neural networks.

In contrast, for local loss, the concentration of QNTK depends not only on expressibility but also on the specific formula of the observable operator $\hat{O}$. Without loss of generality, we assume that a local Pauli observable  $\hat{O}=\hat{O}_i$ acts non-trivially on the $i$-th physical site. Consequently, the upper bound of the trace term $\operatorname{tr}[\hat{O}^2]^2$ in Ineqs.~(\ref{QNTK_Distribution}-\ref{QNTK_Distribution_Local}) generally scales with the order of $2^{2n}$. By strategically designing the variational circuit and local observables, it becomes possible to mitigate, to some extent, the exponential concentration issue arising from high expressibility. As a result, the QNTK matrix can be experimentally estimated in a more efficient manner.

\begin{figure}
\centering
\includegraphics[width=\linewidth]{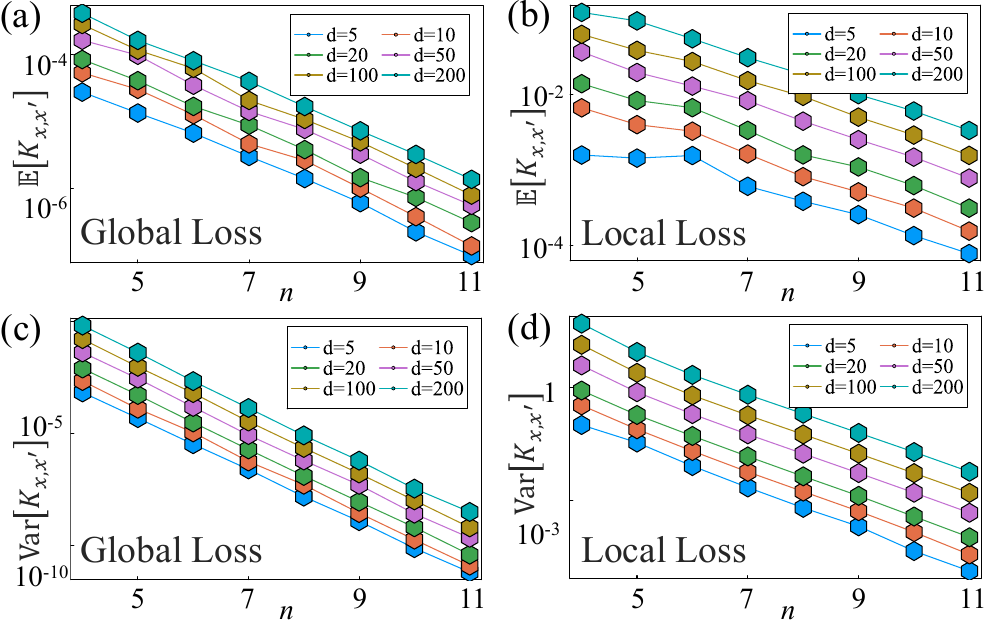} 
\caption{ {Numerical results on the mean and variance of QNTK values for different variational circuit layers.} Here we choose the global quantum encoding. 
(a) and (c) show the mean value and variance of $K(\vec{x},\vec{x}')$ for the global loss function, whereas (b) and (d) show the mean value and variance of $K(\vec{x},\vec{x}')$ for the local loss function. In both cases, the system size $n$ varies from four qubits to eleven qubits, and the number of  variational layer $d$ varies from 5 to 200.}
\label{Fig3_VM}
\end{figure}

{\it Numerical results.}--- To validate our analytical theories regarding the scaling of QNTK concentration concerning quantum system size, we conduct numerical experiments using the open-source package {\tt Yao.jl} \cite{Luo2020Yaojlextensible} in the Julia programming language. Our primary focus is on the global encoding strategy, while additional numerical details on the local encoding case are available in the Supplemental Materials \cite{EICQNTK_supplementary}. We consider mean square loss functions of the type in Eq.~(\ref{MSE_Loss}), encompassing both global and local cases. For global loss, we select the observable $\hat{O}=(|0\rangle\langle0|)^{\otimes n}$, and for the local loss, we use $\hat{O}=Y_1$, where $Y_1$ represents the Pauli $Y$ operator acting on the first qubit site. Our analysis concentrate on highly expressive encoding ensembles $\mathbb{V}$ with $\mathcal{M}^{[1,2]}_{\mathbb{V}}\rightarrow0$. The numerical results, as shown in Figs.~\ref{Fig2_var}(c-d), reveal that both the mean and variance of QNTK over the input data space decay exponentially for both local and global loss, with a relatively modest system size ranging from four to eleven qubits. To support our findings, we conduct numerical experiments with various loss functions and varying variational layers. Fig.~\ref{Fig3_VM} demonstrates that different learning models exhibit exponential concentration of QNTK with similar decay rates, irrespective of the number $d$ of variational layers. 
Simultaneously, Fig.~\ref{Fig3_VM} illustrates that the mean and variance of QNTK increase as the number $d$ of variational layers improves for a constant system size $n$. This aligns with the analytical results presented in Theorem~\ref{theorem:global}, where the number of variational parameters $\Lambda$ scales proportionally with $d$. However, it is challenging to efficiently mitigate the exponential concentration of QNTK by solely increasing the number $d$ of training layers, as $d$ cannot scale exponentially with respect to the system size $n$.

It is worth noting that the decay rate of the mean and variance values for local loss is considerably smaller than that for global loss, as evident in Figs.~\ref{Fig2_var}-\ref{Fig3_VM}. This discrepancy is reasonable since the term $\operatorname{tr}[ \hat{O}^2]$ in Ineq.~(\ref{QNTK_Distribution}) equals 1 for the global observable $\hat{O}=(|0\rangle\langle 0|)^{\otimes n}$, while the term $\operatorname{tr}[ \hat{O}^2]\sim 2^n$ for the local observable $\hat{O}=Y_1$. Consequently, by judiciously selecting a local loss function, it is possible to partially alleviate the issue of exponential concentration. These numerical results align precisely with the analytical findings in Theorem~\ref{theorem:global} and showcase the importance of selecting appropriate loss functions when using QNTK to describe the training dynamics of infinite-width quantum neural networks.

\textit{Discussion and Conclusion.}--- The QNTK plays a pivotal role in elucidating the training dynamics of infinite-width quantum neural networks. 
Although we mainly focus on the expressibility-induced concentration of QNTK, there should exist multiple strategies that can induce the exponential concentration issue, including the model initialization, entanglement, noise, {\it etc.}
In a very recent work \cite{Liu2023Analytic}, the authors demonstrate that both the QNTK and the quantum meta-kernel of a randomly initialized quantum circuit model exhibit exponential concentrations in variational parameter space, resulting in the exponential decay of residual training errors. Furthermore, it would be important to investigate the concentration of QNTK for various quantum circuit architectures, such as quantum convolutional neural networks \cite{Cong2019Quantum}, which exhibit immunity to the barren plateau problem \cite{Pesah2021Absence}. 
In the future, it would be interesting and desirable to investigate the concentration properties of the QNTK along this line. 

Another important topic is to investigate the connections between the expressibility of quantum encoding and the learnability of quantum neural network models. A recent work  \cite{Simon2023The} shows that the eigenvalues of classical neural tangent kernels for the infinite-width neural network can predict the learnability of the corresponding models. More specifically, kernel eigenfunctions with larger eigenvalues are more learnable than those with smaller eigenvalues, thus a target function is more learnable if it places more weight on higher eigenvalue modes. In the future, it is highly desirable to extend these results to their quantum counterparts. Then by studying the connections between the expressibility and the eigenvalue distribution of QNTK, one can establish the straightforward and fundamental connections between the expressibility and learnability of wide quantum neural networks.

In conclusion, we have rigorously established the connections between the expressibility of quantum feature maps and the concentration of QNTK in the input data space, employing the theory of unitary design. Particularly, for global loss functions, we prove that highly expressive quantum encoding can induce the exponential concentration of quantum tangent kernel values to zero, hindering the practical discrimination of different QNTK matrix elements in an efficient manner \cite{Thanasilp2022Exponential}. On the other hand, for local loss functions, the corresponding QNTK continues to exhibit exponential concentration due to high expressibility. However, the decay rate can be partially mitigated. We then conduct numerical experiments to validate the validity of our analytical theories for modest system sizes and various variational layer models. These results reveal a new consequence of quantum expressibility in predicting quantum learning dynamics, which would benefit further studies on wide quantum neural networks, both in theory and practical applications.

We thank Xing-Yan Fan and Dian Zhu for helpful discussions. This work was supported by the National Natural Science Foundation of China (Grants No. 12375060, 12075128, and T2225008), the Tsinghua University Dushi Program, the Innovation Program for Quantum Science and Technology (Grant No. 2021ZD0302203), and Shanghai Qi Zhi Institute.

 \bibliography{DengQAIGroup,BP_ML_MPS}

\clearpage
\onecolumngrid
\makeatletter
\setcounter{page}{1}
\setcounter{figure}{0}
\setcounter{equation}{0}
\setcounter{theorem}{0}

\setcounter{secnumdepth}{3}

\makeatletter
\renewcommand{\thefigure}{S\@arabic\c@figure}
\renewcommand \theequation{S\@arabic\c@equation}
\renewcommand \thetable{S\@arabic\c@table}

\begin{center} 
	{\large \bf Supplemental Materials: 
	Expressibility-induced Concentration of  Quantum Neural Tangent Kernels}
\end{center}

\renewcommand\thesection{S\Roman{section}}
{
  \hypersetup{linkcolor=blue}
  \tableofcontents
}

\newpage

\section{Techniques associate to the proofs}
\subsection{Introduction to the unitary t-design}
We begin by introducing the fundamental concepts related to the unitary $t$-design. For any arbitrary measure $dU$ on the unitary group $U(q)$, the $t$-th moments $M_t(dU)$ of $dU$ are defined using the integral formula \cite{Renes2004Symmetric,Dankert2009Exact,Harrow2009Random}:
\begin{equation}
M_t(dU) = \int_{U(q)} dU \prod_{\lambda=1}^t U_{i_\lambda,j_\lambda}\bar{U}_{i'_\lambda,j'_\lambda},
\end{equation}
Here, $U_{i,j}$ represents the $(i,j)$-th element of the unitary matrix $U$, and $\bar{U}$ represents the complex conjugate of $U$.

 The measure $dU$ is called the unitary $t$-design, if the $k$-th moment with respect to $dU$ equals to the $k$-th moment with respect to the Haar measure $dU_H$ for all the positive integers $k\leq t$,
\begin{equation}
    M_1(dU) = M_1(dU_H),\quad M_2(dU) = M_2(dU_H),\,\, \cdots,\,\, M_t(dU) = M_t(dU_H).  
\end{equation}

Based on the theory of the unitary $t$-design, here we calculate the distribution of the values of the quantum neural tangent kernel (QNTK), which can be utilized for describing the training dynamics of wide quantum neural networks. Here we mainly focus on the mean value and variance of the QNTK with respect to the  input data space, thus we only care about the 
1-moment and 2-moment of the random unitary matrices. We consider the case of highly expressive quantum encoding with $\mathcal{A}^{[1,2]}(\cdot)\rightarrow0$, where the quantum encoding layers are approximate unitary designs.

To calculate the average and variance of the QNTK with respect to the whole parameter space, we only need to care about the 1-moment and 2-moment of the random unitary matrices. It has been shown that the random unitary matrices form the approximate unitary $2$-design \cite{Harrow2009Random}.  The approximate unitary 2-design indicates that the first and second moments are approximately the same as the corresponding moments of $U(q)$ over the Haar measure $dU_H$, i.e., $M_1(dU)=M_1(dU_H)$, $M_2(dU)=M_2(dU_H)$. The first and second moments over the Haar measure are given by the Weingarten functions with the following formula
\begin{eqnarray}
&& M_1(dU_H)=\int_{U(q)}dU_H\;U_{l_0,r_0}\bar{U}_{l_0',r_0'}=\frac{1}{N}\delta_{l_0l_0'}\delta_{r_0r_0'},\label{1-design}\\
&&\begin{aligned}M_2(dU_H)
   =&\int_{U(q)}dU_H\;U_{l_0,r_0}U_{l_1,r_1}{\bar{U}_{l_0',r_0'}}\bar{U}_{l_1',r_1'},\\
=& \frac{1}{q^2-1}(\delta_{l_0l_0'}\delta_{l_1l_1'}\delta_{r_0r_0'}\delta_{r_1r_1'}+\delta_{l_0l_1'}\delta_{l_1l_0'}\delta_{r_0r_1'}\delta_{r_1r_0'})\\
& -\frac{1}{q(q^2-1)}(\delta_{l_0l_0'}\delta_{l_1l_1'}\delta_{r_0r_1'}\delta_{r_1r_0'}
 +\delta_{l_0l_1'}\delta_{l_1l_0'}\delta_{r_0r_0'}\delta_{r_1r_1'}), 
\end{aligned}\label{2-design}
\end{eqnarray}
where $U_{l,r}$ represents the left and right dangling leg of the unitary tensor $U$, and $\bar{U}_{l,r}$ represents the complex conjugation of $U_{l,r}$. 

\subsection{Expressibility of the quantum encoding}
\subsubsection{Expressibility of global quantum encoding}
In dealing with classical data set through quantum machine learning models,  one needs to embed the classical input data into the quantum hardwares. One straightforward approach is to encode the classical data into quantum states through quantum feature map. Common approaches include quantum amplitude encoding, time-evolution encoding, Hamiltonian encoding, and so on. In our later numerical simulations, we choose the quantum amplitude encoding to map the classical input samples into quantum states.  

In the main manuscript, we briefly introduce that the $q$-dimensional  classical input data can be encoded into the quantum state $|{\vec{x}}\rangle = V({\vec{x}})|\psi_0\rangle$ through the mapping $\mathbb{V}_{\vec{x}}: {\vec{x}}\in  \mathbb{C}^q \longmapsto V({\vec{x}})\in \mathcal{V}(q)$,  where  $\mathcal{V}(q)$ is the $q$-dimensional  unitary group.  Thus all of the feasible unitary elements $\mathbb{V}_{\vec{x}}$  constitute an ensemble $\mathbb{V}=\{\mathbb{V}_{\vec{x}}\}$. Following the previous work associated to the expressibility of quantum encoding, one can define the   expressibility of an ensemble $\mathbb{V}$ by how the ensemble $\mathbb{V}$ deviates from a unitary $t$-design  \cite{Sim2019Expressibility,Holmes2022Connecting,Thanasilp2022Exponential}
\begin{equation}\label{Supp_Expressibility}
\mathcal{A}^{[t]}_{\mathbb{V}} (\cdot)=\mathcal{D}^{[t]}_{\rm Haar}(\cdot)-\int_{\mathbb{V}}dVV^{\otimes t}(\cdot)^{\otimes t}(V^{\dag})^{\otimes t},
\end{equation}
where $\mathcal{D}^{[t]}_{\rm Haar}(\cdot) = \int_{\mathcal{U}(q)}d\mu(U) U^{\otimes t}(\cdot)^{\otimes t}(U^{\dag})^{\otimes t}$ represents the $t$-moment integral over Haar ensemble of the unitary group $\mathcal{U}(q)$ \cite{Collins2006Integration}, $\mathcal{A}^{[t]}_{\mathbb{V}}$ is a multi-dimensional tensor and denotes the $t$-moment expressibility. The label ``$\cdot$'' in Eq.~(\ref{Supp_Expressibility}) corresponds to the initial quantum density matrix to be prepared in the quantum learning system.  

Then for a given initial density matrix $\rho_0$, one can define the measure of the $t$-moment expressibility in terms of the trace norm \cite{Thanasilp2022Exponential}, 
\begin{equation}\label{Supp_Measure}
\mathcal{M}^{[t]}_{\mathbb{V}} := \left \| \mathcal{A}^{[t]}_{\mathbb{V}}(\rho_0)\right \| _1.
\end{equation}
As a real number, the expressibility measure $\mathcal{M}^{[t]}_{\mathbb{V}}$ is more convenient in representing the degree of expressibility of arbitrary quantum encoding circuits. For $t$-moment maximal expressibility, the term $\mathcal{M}^{[t]}_{\mathbb{V}}=0$. Besides, as has been mentioned in the main text,  the unitary $t$-design is downward compatible, i.e.,  if an ensemble forms the unitary $t$-design, it also forms the unitary $(t-1)$-design. Thus the expressibility measure of an ensemble is also downward compatible, i.e., $\mathcal{M}^{[t]}_{\mathbb{V}}\rightarrow0$ indicates $\mathcal{M}^{[t-1]}_{\mathbb{V}}\rightarrow0$.  

\subsubsection{Expressibility of local quantum encoding}
Now we consider the local quantum encoding through the mapping $\mathbb{W}_{\vec{x}}: {\vec{x}}\longmapsto W({\vec{x}}) =\bigotimes_{i=1}^{n} W_i (\vec{x})$ , where $W_i \in {\rm SU}(2)$ locates on the $i$-th  qubit site.   

One can quantify the  expressibility of the local tensor product ensemble $\mathbb{W}$ by how the ensemble $\mathbb{W}$ deviates from the  $t$-design of local tensor product unitary group 
\begin{equation}\label{Supp_Expressibility_Local}
\mathcal{A}^{[t]}_{\mathbb{W}} (\cdot)=\mathcal{V}^{[t]}_{\rm Haar^{\otimes n}}(\cdot)-\int_{\mathbb{W}}dWW^{\otimes t}(\cdot)^{\otimes t}(W^{\dag})^{\otimes t},
\end{equation}
where $\mathcal{V}^{[t]}_{\rm Haar^{\otimes n}}(\cdot) = \int_{\mathcal{U}}d\mu(U) U^{\otimes t}(\cdot)^{\otimes t}(U^{\dag})^{\otimes t}$ represents the $t$-moment integral over the Haar ensemble of the local tensor product unitary group $\mathcal{U} = \bigotimes_{i=1}^n \mathcal{U}_i$. $\mathcal{A}^{[t]}_{\mathbb{W}}$ denotes the $t$-moment expressibility of local tensor product unitary ensembles $\mathbb{W}$. The label ``$\cdot$'' in Eq.~(\ref{Supp_Expressibility_Local}) corresponds to the initial quantum density matrix to be prepared in the quantum learning system.  

Then for a given initial density matrix $\rho_0$, one can define the measure of the $t$-moment local expressibility in terms of the trace norm \cite{Thanasilp2022Exponential}, 
\begin{equation}\label{Supp_Measure_Local}
\mathcal{M}^{[t]}_{\mathbb{W}} := \left \| \mathcal{A}^{[t]}_{\mathbb{W}}(\rho_0)\right \| _1.
\end{equation}

As was mentioned above, here we calculate the expectation value and variance of  QNTK matrix elements, and these calculations require at most the 2-moment term of unitary group integrals. Thus we only need to consider the 1-moment and 2-moment  expressibilities  $\mathcal{A}^{[1,2]}_{\mathbb{W}}$, as well as their measures $\mathcal{M}^{[1,2]}_{\mathbb{W}}$.

%

\section{Proof of Theorem 1: Global expressibility induced concentration of QNTK }
Now we investigate the probability distribution of the QNTK values over the input data space. We obtain the following theorem, which has been presented in the main text. 

\begin{theorem}\label{Supp_theorem:global} 
Given the above quantum neural tangent kernel defined in Eq.~(\ref{QNTK}). 
Then for  arbitrary two input samples $\vec{x}$ and $\vec{x}'$ drawn from the same distribution, the QNTK  $K(\vec{x}, \vec{x}')$ obeys
\begin{equation}
\begin{aligned}
{\rm Pr}_{\vec{x}, \vec{x}'}\left(|K(\vec{x}, \vec{x}')+\mathcal{O}(\mathcal{M}^{[1]}_{\mathbb{V}})|\geq \epsilon\right)
\leq  \frac{\epsilon^{-2}\Lambda^2}{(2^{2n}-1)^2}\left(2 {{\rm tr}}[ \hat{O}^2]\right)^2-\epsilon^{-2}\mathcal{O}(\mathcal{M}^{[1]}_{\mathbb{V}},\mathcal{M}^{[2]}_{\mathbb{V}}),
\end{aligned}
\end{equation}
for arbitrary small positive constant $\epsilon$, where $n$ denotes the number of qubits, $\Lambda$ denotes the number of trainable parameters, and $\mathcal{M}^{[1,2]}_{\mathbb{V}}$ means the measure of the $t$-moment expressibility of unitary ensemble $\mathbb{V}$.
\end{theorem}

Now we present the proof details of Theorem 1. 

We adopt the Chebyschev's inequality
\begin{equation}\label{Chebyschev_Inequality}
\textrm{Pr}(|X-\mathbb{E}[X]|\geq \epsilon)\leq \frac{{\rm Var}[X]}{\epsilon^2}, \, (\epsilon>0)
\end{equation}
to estimate the value distribution over the data space, where $\mathbb{E}[X]$ and ${\rm Var}[X]$ represent the expectation value and variance of $X$, respectively. In the following subsections,  we mainly focus on deriving the analytical formula of  the expectation value $\mathbb{E}[K_{\vec{x}, \vec{x}'}]$ and variance ${\rm Var}[K_{\vec{x}, \vec{x}'}]$  of the quantum tangent kernel matrix elements with respect to the input data set.

\subsection{Expectation of the QNTK over the  globally encoded  space}

We first calculate the expectation value of the QNTK with respect to the input data space. Here $N$ denotes the dimension of $n$-qubit Hilbert space, $N\equiv2^n$. The expectation value of $K_{\vec{x},  \vec{x}'}$ takes the form
\begin{equation}\label{Supp_Expectation_Global}
\begin{aligned}
\mathbb{E}[K_{\vec{x},  \vec{x}'}]&=\sum_{\lambda} \int_{\mathbb{V}} dV(\vec{x})dV(\vec{x}')\tr \left[ V^\dag(\vec{x})\partial_{\theta_\lambda} \hat{O}(\theta)V(\vec{x})\rho_0 V^\dag(\vec{x}')\partial_{\theta_\lambda} \hat{O}(\theta)V(\vec{x}')\rho_0\right]\\
&=\sum_{\lambda} \int_{\mathbb{V}} dV(\vec{x})dV(\vec{x}')\tr \left[ V^\dag(\vec{x})_{ij}(\partial_{\theta_\lambda} \hat{O}(\theta))_{jk}V(\vec{x})_{kl}(\rho_0)_{lm} V^\dag(\vec{x}')_{mn}(\partial_{\theta_\lambda} \hat{O}(\theta))_{no}V(\vec{x}')_{op}(\rho_0)_{pq}\right]\\
&=\sum_{\lambda} \tr\left[\left(\frac{1}{N}\delta_{il}\delta_{jk}-(\mathcal{A}^{[1]}_{\mathbb{V}})_{il,jk}\right)(\partial_{\theta_\lambda} \hat{O}(\theta))_{jk}(\rho_0)_{lm}\left(\frac{1}{N}\delta_{mp}\delta_{no}-(\mathcal{A}^{[1]}_{\mathbb{V}})_{mp,no}\right)(\partial_{\theta_\lambda} \hat{O}(\theta))_{no}(\rho_0)_{pq}\right]\\
&=\sum_{\lambda}\left[\left(\frac{1}{N}\delta_{il}\delta_{jk}-(\mathcal{A}^{[1]}_{\mathbb{V}})_{il,jk}\right)(\partial_{\theta_\lambda} \hat{O}(\theta))_{jk}(\rho_0)_{lm}\left(\frac{1}{N}\delta_{mp}\delta_{no}-(\mathcal{A}^{[1]}_{\mathbb{V}})_{mp,no}\right)(\partial_{\theta_\lambda} \hat{O}(\theta))_{no}(\rho_0)_{pi}\right]\\
&=\sum_{\lambda} \left[N^{-2}\delta_{il}\delta_{jk}\delta_{mp}\delta_{no}-N^{-1}\delta_{il}\delta_{jk}(\mathcal{A}^{[1]}_{\mathbb{V}})_{mp,no}-N^{-1}\delta_{mp}\delta_{no}(\mathcal{A}^{[1]}_{\mathbb{V}})_{il,jk}+(\mathcal{A}^{[1]}_{\mathbb{V}})_{il,jk}(\mathcal{A}^{[1]}_{\mathbb{V}})_{mp,no}\right]\\
& \qquad \quad \times (\partial_{\theta_\lambda} \hat{O}(\theta))_{jk}(\rho_0)_{lm}(\partial_{\theta_\lambda} \hat{O}(\theta))_{no}(\rho_0)_{pi}\\
&= \frac{1}{N^2} \sum_\lambda (\partial_{\theta_\lambda} \hat{O}(\theta))_{jj}(\partial_{\theta_\lambda} \hat{O}(\theta))_{nn}(\rho_0)_{ip}(\rho_0)_{pi}-\mathcal{O}(\mathcal{M}^{[1]}_{\mathbb{V}})\\
&=  \frac{1}{N^2} \sum_\lambda \left(\tr\left[\partial_{\theta_\lambda} \hat{O}(\theta)\right]\right)^2-\mathcal{O}(\mathcal{M}^{[1]}_{\mathbb{V}})\\
&=  \frac{1}{2^{2n}} \sum_\lambda \left(\tr\left[\partial_{\theta_\lambda} \hat{O}(\theta)\right]\right)^2-\mathcal{O}(\mathcal{M}^{[1]}_{\mathbb{V}}),
\end{aligned}
\end{equation}
where $\mathcal{M}^{[1]}_{\mathbb{V}}$ denotes the measure of 1-moment expressibility, see Eq.~(\ref{Supp_Measure}).

Actually, for 1-moment maximal expressibility with $\mathcal{M}^{[1]}_{\mathbb{V}}=0$,  we see from Eq.~(\ref{Supp_Expectation_Global}) that the expectation value tends to be exponentially small with respect to the qubit number of system, as long as the term $\sum_\lambda \left(\tr\left[\partial_{\theta_\lambda} \hat{O}(\theta)\right]\right)^2$ does not grow exponentially with respect to the qubit number $n$.   Without loss of generality, here we consider a special case of the variational circuit $U(\vec{\theta})$, which has been widely adopted in literature \cite{Mcclean2018Barren},
\begin{equation}\label{U_parametrize}
U(\vec{\theta}) = \prod_{j=1} \exp (i\theta_jW_j),\quad (\forall j\quad   W_j=W_j^\dag, \,W_j^2=1).
\end{equation}
Through concrete calculations, one can then obtain that the trace of the gradient term vanishes, i.e., 
\begin{equation}
\tr[\partial_{\theta_\lambda} \hat{O}(\theta)]=\tr[\partial_{\theta_\lambda} (U^\dag(\theta)\hat{O}U(\theta))]=0.
\end{equation}
Therefore we have the expectation value of $K_{\vec{x}, \vec{x}'}$
\begin{equation}\label{Supp_Expectation}
\mathbb{E}[K_{\vec{x},  \vec{x}'}]= -\mathcal{O}(\mathcal{M}^{[1]}_{\mathbb{V}}) = 0,
\end{equation}
for maximal expressibility of the unitary ensemble $\mathbb{V}$.

\subsection{Variance of the QNTK over the  globally encoded  space}
Now we calculate the variance of the quantum tangent kernel over the input data space
\begin{equation}
\begin{aligned}
{\rm Var}[K_{\vec{x},  \vec{x}'}] &= \mathbb{E}[K_{\vec{x},  \vec{x}'}^2] - \left(\mathbb{E}[K_{\vec{x},  \vec{x}'}]\right)^2= \mathbb{E}[K_{\vec{x},  \vec{x}'}^2] -\mathcal{O}((\mathcal{M}^{[1]}_{\mathbb{V}})^2)\\
&= \sum_{\lambda,\lambda'} \int_{\mathbb{V}} dV(\vec{x})dV(\vec{x}')\tr\left[V^\dag(\vec{x})\partial_{\theta_\lambda} \hat{O}(\theta)V(\vec{x})\rho_0 V^\dag(\vec{x}')\partial_{\theta_\lambda} \hat{O}(\theta)V(\vec{x}')\rho_0\right] \\
&\cdot \tr\left[V^\dag(\vec{x})\partial_{\theta_{\lambda'}} \hat{O}(\theta)V(\vec{x})\rho_0 V^\dag(\vec{x}')\partial_{\theta_{\lambda'}} \hat{O}(\theta)V(\vec{x}')\rho_0\right]-\mathcal{O}((\mathcal{M}^{[1]}_{\mathbb{V}})^2)
\nonumber
\end{aligned}
\end{equation}
\begin{equation} \label{Supp_Var}
\begin{aligned}
&= \sum_{\lambda,\lambda'}\tr \left[ \underbrace{\int_{\mathbb{V}} dV(\vec{x}) (V^\dag(\vec{x}))^{\otimes2} (\partial_\lambda\otimes \partial_{\lambda'})(V(\vec{x}))^{\otimes2}}_\textrm{ \ding{172}}\rho_0^{\otimes2} \underbrace{\int_{\mathbb{V}} dV(\vec{x}') (V^\dag(\vec{x}'))^{\otimes2} (\partial_\lambda\otimes \partial_{\lambda'})(V(\vec{x}'))^{\otimes2}}_\textrm{ \ding{173}}\rho_0^{\otimes2}\right]\\
&\quad -\mathcal{O}\left(\left(\mathcal{M}^{[1]}_{\mathbb{V}}\right)^2\right),
\end{aligned}
\end{equation}
where the term $\partial_{\theta_\lambda} \hat{O}(\theta)$ is abbreviated as $\partial_\lambda$ for latter convenience. We first calculate the term- \ding{172} in Eq.~(\ref{Supp_Var}),
\begin{equation}
\begin{aligned}
\textrm{\ding{172}}=& \int_{\mathbb{V}} dV(\vec{x}) (V^\dag(\vec{x}))^{\otimes2} (\partial_\lambda\otimes \partial_{\lambda'})(V(\vec{x}))^{\otimes2}\\
=& \int_{\mathbb{V}} dV(\vec{x}) (V^\dag(\vec{x}))^{\otimes2}_{ii',jj'} (\partial_\lambda\otimes \partial_{\lambda'})_{jj',kk'}(V(\vec{x}))^{\otimes2}_{kk',ll'}\\
=& \frac{1}{N^2-1}\left[\left(\underbrace{\delta_{jk}\delta_{j'k'}\delta_{il}\delta_{i'l'}}_\textrm{ \ding{174}}+\underbrace{\delta_{jk'}\delta_{j'k}\delta_{il'}\delta_{i'l}}_\textrm{ \ding{175}}\right)-\frac{1}{N}\left(\underbrace{\delta_{jk}\delta_{j'k'}\delta_{il'}\delta_{i'l}}_\textrm{ \ding{176}}+\underbrace{\delta_{jk'}\delta_{j'k}\delta_{il}\delta_{i'l'}}_\textrm{ \ding{177}}\right)\right]\cdot(\partial_\lambda\otimes \partial_{\lambda'})_{jj',kk'}-\left(\mathcal{A}^{[2]}_{\mathbb{V}}\right)_{ii',ll'}.
\end{aligned}
\end{equation} 
Similarly, the term-\ding{173} in Eq.~(\ref{Supp_Var}) reads
\begin{equation}
\begin{aligned}
\textrm{\ding{173}}=& \int_{\mathbb{V}} dV(\vec{x}') (V^\dag(\vec{x}'))^{\otimes2} (\partial_\lambda\otimes \partial_{\lambda'})(V(\vec{x}'))^{\otimes2}\\
=& \int_{\mathbb{V}} dV(\vec{x}') (V^\dag(\vec{x}'))^{\otimes2}_{mm',nn'} (\partial_\lambda\otimes \partial_{\lambda'})_{nn',oo'}(V(\vec{x}'))^{\otimes2}_{oo',pp'}\\
=& \frac{1}{N^2-1}\left[\left(\underbrace{\delta_{no}\delta_{n'o'}\delta_{mp}\delta_{m'p'}}_\textrm{ \ding{178}}+\underbrace{\delta_{no'}\delta_{n'o}\delta_{mp'}\delta_{m'p}}_\textrm{ \ding{179}}\right)-\frac{1}{N}\left(\underbrace{\delta_{no}\delta_{n'o'}\delta_{mp'}\delta_{m'p}}_\textrm{ \ding{180}}+\underbrace{\delta_{no'}\delta_{n'o}\delta_{mp}\delta_{m'p'}}_\textrm{ \ding{181}}\right)\right]\cdot(\partial_\lambda\otimes \partial_{\lambda'})_{nn',oo'}\\
&-\left(\mathcal{A}^{[2]}_{\mathbb{V}}\right)_{mm',pp'}.
\end{aligned}
\end{equation} 
Then the variance of $K_{\bf x,x'}$
\begin{equation}
\begin{aligned}
{\rm Var}[K_{\vec{x}, \vec{x}'}] &= \mathbb{E}[K_{\vec{x}, \vec{x}'}^2] - \left(\mathbb{E}[K_{\vec{x}, \vec{x}'}]\right)^2= \mathbb{E}[K_{\vec{x}, \vec{x}'}^2]-\mathcal{O}((\mathcal{M}^{[1]}_{\mathbb{V}})^2) \\
&= \sum_{\lambda,\lambda'} \int_{\mathbb{V}} dV(\vec{x})dV(\vec{x}')\tr\left[V^\dag(\vec{x})\partial_{\theta_\lambda} \hat{O}(\theta)V(\vec{x})\rho_0 V^\dag(\vec{x}')\partial_{\theta_\lambda} \hat{O}(\theta)V(\vec{x}')\rho_0\right]\\
&\qquad \cdot \tr\left[V^\dag(\vec{x})\partial_{\theta_{\lambda'}} \hat{O}(\theta)V(\vec{x})\rho_0 V^\dag(\vec{x}')\partial_{\theta_{\lambda'}} \hat{O}(\theta)V(\vec{x}')\rho_0\right]-\mathcal{O}((\mathcal{M}^{[1]}_{\mathbb{V}})^2)\\
&= \sum_{\lambda,\lambda'}\tr \left[ \underbrace{\int_{\mathbb{V}} dV(\vec{x}) (V^\dag(\vec{x}))^{\otimes2} (\partial_\lambda\otimes \partial_{\lambda'})(V(\vec{x}))^{\otimes2}}\rho_0^{\otimes2} \underbrace{\int_{\mathbb{V}} dV(\vec{x}') (V^\dag(\vec{x}'))^{\otimes2} (\partial_\lambda\otimes \partial_{\lambda'})(V(\vec{x}'))^{\otimes2}}\rho_0^{\otimes2}\right]\\
&\quad -\mathcal{O}((\mathcal{M}^{[1]}_{\mathbb{V}})^2)\\
&=  \frac{1}{(N^2-1)^2}\sum_{\lambda,\lambda'}\tr \left[\left((\textrm{\ding{174}}+\textrm{\ding{175}})-\frac{1}{N}(\textrm{\ding{176}}+\textrm{\ding{177}})\right)\left((\textrm{\ding{178}}+\textrm{\ding{179}})-\frac{1}{N}(\textrm{\ding{180}}+\textrm{\ding{181}})\right)\right.\\
&\qquad \qquad \qquad  \qquad  \times \left.(\partial_\lambda\otimes \partial_{\lambda'})_{jj',kk'}(\partial_\lambda\otimes \partial_{\lambda'})_{nn',oo'}(\rho_0)_{ll',mm'}(\rho_0)_{pp',qq'}\right] +\mathcal{O}((\mathcal{M}^{[1]}_{\mathbb{V}})^2,\mathcal{M}^{[2]}_{\mathbb{V}}).
\end{aligned}
\end{equation}
There are totally 16 terms $\{${\ding{174}}, {\ding{175}}, {\ding{176}}, {\ding{177}}$\}\times\{${\ding{178}}, {\ding{179}}, {\ding{180}}, {\ding{181}}$\}$, 
\begin{equation}
\begin{aligned}
\textrm{\ding{174}}\times\textrm {\ding{178}} & \rightarrow  \frac{1}{(N^2-1)^2}\sum_{\lambda,\lambda'}\tr \left[\delta_{jk}\delta_{j'k'}\delta_{il}\delta_{i'l'}\delta_{no}\delta_{n'o'}\delta_{mp}\delta_{m'p'}(\partial_\lambda\otimes \partial_{\lambda'})_{jj',kk'}(\partial_\lambda\otimes \partial_{\lambda'})_{nn',oo'}(\rho_0^{\otimes2})_{ll',mm'}(\rho_0^{\otimes2})_{pp',qq'}\right]\\
& =  \frac{1}{(N^2-1)^2}\sum_{\lambda,\lambda'} \tr\left[ (\partial_\lambda\otimes \partial_{\lambda'})_{jj',jj'}(\partial_\lambda\otimes \partial_{\lambda'})_{nn',nn'}(\rho_0^{\otimes2})_{ii',mm'}(\rho_0^{\otimes2})_{mm',qq'}\right]\\
& = \frac{1}{(N^2-1)^2}\sum_{\lambda,\lambda'} \left(\tr \left[\partial_\lambda\otimes \partial_{\lambda'}\right]\right)^2\tr[(\rho_0^2)^{\otimes2}]\\
& = \frac{1}{(N^2-1)^2}\sum_{\lambda,\lambda'} \left(\tr \left[\partial_\lambda\right]\tr \left[ \partial_{\lambda'}\right]\tr\left[\rho_0^2\right]\right)^2.
\end{aligned}
\end{equation}

\begin{equation}
\begin{aligned}
\textrm{\ding{174}}\times\textrm {\ding{179}} & \rightarrow  \frac{1}{(N^2-1)^2}\sum_{\lambda,\lambda'}\tr \left[\delta_{jk}\delta_{j'k'}\delta_{il}\delta_{i'l'}\delta_{no'}\delta_{n'o}\delta_{mp'}\delta_{m'p}(\partial_\lambda\otimes \partial_{\lambda'})_{jj',kk'}(\partial_\lambda\otimes \partial_{\lambda'})_{nn',oo'}(\rho_0^{\otimes2})_{ll',mm'}(\rho_0^{\otimes2})_{pp',qq'}\right]\\
& =  \frac{1}{(N^2-1)^2}\sum_{\lambda,\lambda'} \tr\left[ (\partial_\lambda\otimes \partial_{\lambda'})_{jj',jj'}(\partial_\lambda\otimes \partial_{\lambda'})_{nn',n'n}(\rho_0^{\otimes2})_{ii',mm'}(\rho_0^{\otimes2})_{m'm,qq'}\right]\\
& = \frac{1}{(N^2-1)^2}\sum_{\lambda,\lambda'} \left(\tr \left[\partial_\lambda\otimes\partial_{\lambda'}\right]\tr \left[\partial_\lambda \otimes\partial_{\lambda'} {\rm SWAP}\right]\tr \left[(\rho_0^2)^{\otimes2}{\rm SWAP}\right]\right)\\
& = \frac{1}{(N^2-1)^2}\sum_{\lambda,\lambda'} \left(\tr \left[\partial_\lambda\right]\tr \left[\partial_{\lambda'}\right]\tr \left[\partial_\lambda \partial_{\lambda'}\right]\tr \left[\rho_0^4\right]\right)
\end{aligned}
\end{equation}

\begin{equation}
\begin{aligned}
\textrm{\ding{174}}\times\textrm {\ding{180}} & \rightarrow  -\frac{1}{N(N^2-1)^2}\sum_{\lambda,\lambda'}\tr \left[\delta_{jk}\delta_{j'k'}\delta_{il}\delta_{i'l'}\delta_{no}\delta_{n'o'}\delta_{mp'}\delta_{m'p}(\partial_\lambda\otimes \partial_{\lambda'})_{jj',kk'}(\partial_\lambda\otimes \partial_{\lambda'})_{nn',oo'}(\rho_0^{\otimes2})_{ll',mm'}(\rho_0^{\otimes2})_{pp',qq'}\right]\\
& = - \frac{1}{N(N^2-1)^2}\sum_{\lambda,\lambda'} \tr\left[ (\partial_\lambda\otimes \partial_{\lambda'})_{jj',jj'}(\partial_\lambda\otimes \partial_{\lambda'})_{nn',nn'}(\rho_0^{\otimes2})_{ii',mm'}(\rho_0^{\otimes2})_{m'm,qq'}\right]\\
& =- \frac{1}{N(N^2-1)^2}\sum_{\lambda,\lambda'} \left(\tr \left[\partial_\lambda\otimes\partial_{\lambda'}\right]\tr \left[\partial_\lambda \otimes\partial_{\lambda'} \right]\tr \left[(\rho_0^2)^{\otimes2}{\rm SWAP}\right]\right)\\
& = -\frac{1}{N(N^2-1)^2}\sum_{\lambda,\lambda'} \left((\tr \left[\partial_\lambda\right])^2(\tr \left[\partial_{\lambda'}\right])^2\tr \left[\rho_0^4\right]\right)
\end{aligned}
\end{equation}

\begin{equation}
\begin{aligned}
\textrm{\ding{174}}\times\textrm {\ding{181}} & \rightarrow  -\frac{1}{N(N^2-1)^2}\sum_{\lambda,\lambda'}\tr \left[\delta_{jk}\delta_{j'k'}\delta_{il}\delta_{i'l'}\delta_{no'}\delta_{n'o}\delta_{mp}\delta_{m'p'}(\partial_\lambda\otimes \partial_{\lambda'})_{jj',kk'}(\partial_\lambda\otimes \partial_{\lambda'})_{nn',oo'}(\rho_0^{\otimes2})_{ll',mm'}(\rho_0^{\otimes2})_{pp',qq'}\right]\\
& = - \frac{1}{N(N^2-1)^2}\sum_{\lambda,\lambda'} \tr\left[ (\partial_\lambda\otimes \partial_{\lambda'})_{jj',jj'}(\partial_\lambda\otimes \partial_{\lambda'})_{nn',n'n}(\rho_0^{\otimes2})_{ii',mm'}(\rho_0^{\otimes2})_{mm',qq'}\right]\\
& =- \frac{1}{N(N^2-1)^2}\sum_{\lambda,\lambda'} \left(\tr \left[\partial_\lambda\otimes\partial_{\lambda'}\right]\tr \left[\partial_\lambda \otimes\partial_{\lambda'}{\rm SWAP} \right]\tr \left[(\rho_0^2)^{\otimes2}\right]\right)\\
& = -\frac{1}{N(N^2-1)^2}\sum_{\lambda,\lambda'} \left(\tr \left[\partial_\lambda\right]\tr \left[\partial_{\lambda'}\right]\tr \left[\partial_\lambda \partial_{\lambda'}\right](\tr \left[\rho_0^2\right])^2\right)
\end{aligned}
\end{equation}

\begin{equation}
\begin{aligned}
\textrm{\ding{175}}\times\textrm {\ding{178}} =\textrm{\ding{174}}\times\textrm {\ding{179}} \rightarrow = \frac{1}{(N^2-1)^2}\sum_{\lambda,\lambda'} \left(\tr \left[\partial_\lambda\right]\tr \left[\partial_{\lambda'}\right]\tr \left[\partial_\lambda \partial_{\lambda'}\right]\tr \left[\rho_0^4\right]\right)
\end{aligned}
\end{equation}

\begin{equation}
\begin{aligned}
\textrm{\ding{175}}\times\textrm {\ding{179}} & \rightarrow  \frac{1}{(N^2-1)^2}\sum_{\lambda,\lambda'}\tr \left[\delta_{jk'}\delta_{j'k}\delta_{il'}\delta_{i'l}\delta_{no'}\delta_{n'o}\delta_{mp'}\delta_{m'p}(\partial_\lambda\otimes \partial_{\lambda'})_{jj',kk'}(\partial_\lambda\otimes \partial_{\lambda'})_{nn',oo'}(\rho_0^{\otimes2})_{ll',mm'}(\rho_0^{\otimes2})_{pp',qq'}\right]\\
& =  \frac{1}{(N^2-1)^2}\sum_{\lambda,\lambda'} \tr\left[ (\partial_\lambda\otimes \partial_{\lambda'})_{jj',j'j}(\partial_\lambda\otimes \partial_{\lambda'})_{nn',n'n}(\rho_0^{\otimes2})_{i'i,mm'}(\rho_0^{\otimes2})_{m'm,qq'}\right]\\
& = \frac{1}{(N^2-1)^2}\sum_{\lambda,\lambda'} \left((\tr \left[\partial_\lambda \otimes\partial_{\lambda'} {\rm SWAP}\right])^2\tr \left[(\rho_0^2)^{\otimes2}\right]\right)\\
& = \frac{1}{(N^2-1)^2}\sum_{\lambda,\lambda'} \left(\tr \left[\partial_\lambda \partial_{\lambda'}\right]\tr \left[\rho_0^2\right]\right)^2
\end{aligned}
\end{equation}

\begin{equation}
\begin{aligned}
\textrm{\ding{175}}\times\textrm {\ding{180}} & \rightarrow  -\frac{1}{N(N^2-1)^2}\sum_{\lambda,\lambda'}\tr \left[\delta_{jk'}\delta_{j'k}\delta_{il'}\delta_{i'l}\delta_{no}\delta_{n'o'}\delta_{mp'}\delta_{m'p}(\partial_\lambda\otimes \partial_{\lambda'})_{jj',kk'}(\partial_\lambda\otimes \partial_{\lambda'})_{nn',oo'}(\rho_0^{\otimes2})_{ll',mm'}(\rho_0^{\otimes2})_{pp',qq'}\right]\\
& = - \frac{1}{N(N^2-1)^2}\sum_{\lambda,\lambda'} \tr\left[ (\partial_\lambda\otimes \partial_{\lambda'})_{jj',j'j}(\partial_\lambda\otimes \partial_{\lambda'})_{nn',nn'}(\rho_0^{\otimes2})_{i'i,mm'}(\rho_0^{\otimes2})_{m'm,qq'}\right]\\
& =- \frac{1}{N(N^2-1)^2}\sum_{\lambda,\lambda'} \left(\tr \left[\partial_\lambda\otimes\partial_{\lambda'}{\rm SWAP}\right]\tr \left[\partial_\lambda \otimes\partial_{\lambda'} \right]\tr \left[(\rho_0^2)^{\otimes2}\right]\right)\\
& = -\frac{1}{N(N^2-1)^2}\sum_{\lambda,\lambda'} \left(\tr \left[\partial_\lambda \partial_{\lambda'}\right]\tr \left[\partial_\lambda\right]\tr \left[\partial_{\lambda'}\right](\tr \left[\rho_0^2\right])^2\right)
\end{aligned}
\end{equation}

\begin{equation}
\begin{aligned}
\textrm{\ding{175}}\times\textrm {\ding{181}} & \rightarrow  -\frac{1}{N(N^2-1)^2}\sum_{\lambda,\lambda'}\tr \left[\delta_{jk'}\delta_{j'k}\delta_{il'}\delta_{i'l} \delta_{no'}\delta_{n'o}\delta_{mp}\delta_{m'p'}(\partial_\lambda\otimes \partial_{\lambda'})_{jj',kk'}(\partial_\lambda\otimes \partial_{\lambda'})_{nn',oo'}(\rho_0^{\otimes2})_{ll',mm'}(\rho_0^{\otimes2})_{pp',qq'}\right]\\
& = - \frac{1}{N(N^2-1)^2}\sum_{\lambda,\lambda'} \tr\left[ (\partial_\lambda\otimes \partial_{\lambda'})_{jj',j'j}(\partial_\lambda\otimes \partial_{\lambda'})_{nn',n'n}(\rho_0^{\otimes2})_{i'i,mm'}(\rho_0^{\otimes2})_{mm',qq'}\right]\\
& =- \frac{1}{N(N^2-1)^2}\sum_{\lambda,\lambda'} \left((\tr \left[\partial_\lambda \otimes\partial_{\lambda'}{\rm SWAP} \right])^2\tr \left[(\rho_0^2)^{\otimes2}{\rm SWAP} \right]\right)\\
& = -\frac{1}{N(N^2-1)^2}\sum_{\lambda,\lambda'} \left((\tr \left[\partial_\lambda\partial_{\lambda'}\right])^2\tr \left[\rho_0^4\right]\right)
\end{aligned}
\end{equation}

\begin{equation}
\begin{aligned}
\textrm{\ding{176}}\times\textrm {\ding{178}} = \textrm{\ding{174}}\times\textrm {\ding{180}}  \rightarrow    -\frac{1}{N(N^2-1)^2}\sum_{\lambda,\lambda'} \left((\tr \left[\partial_\lambda\right])^2(\tr \left[\partial_{\lambda'}\right])^2\tr \left[\rho_0^4\right]\right)
\end{aligned}
\end{equation}

\begin{equation}
\begin{aligned}
\textrm{\ding{176}}\times\textrm {\ding{179}} =\textrm{\ding{175}}\times\textrm {\ding{180}}  \rightarrow   -\frac{1}{N(N^2-1)^2}\sum_{\lambda,\lambda'} \left(\tr \left[\partial_\lambda \partial_{\lambda'}\right]\tr \left[\partial_\lambda\right]\tr \left[\partial_{\lambda'}\right](\tr \left[\rho_0^2\right])^2\right)
\end{aligned}
\end{equation}

\begin{equation}
\begin{aligned}
\textrm{\ding{176}}\times\textrm {\ding{180}} & \rightarrow  \frac{1}{N^2(N^2-1)^2}\sum_{\lambda,\lambda'}\tr \left[\delta_{jk}\delta_{j'k'}\delta_{il'}\delta_{i'l}\delta_{no}\delta_{n'o'}\delta_{mp'}\delta_{m'p}(\partial_\lambda\otimes \partial_{\lambda'})_{jj',kk'}(\partial_\lambda\otimes \partial_{\lambda'})_{nn',oo'}(\rho_0^{\otimes2})_{ll',mm'}(\rho_0^{\otimes2})_{pp',qq'}\right]\\
& =  \frac{1}{N^2(N^2-1)^2}\sum_{\lambda,\lambda'} \tr\left[ (\partial_\lambda\otimes \partial_{\lambda'})_{jj',jj'}(\partial_\lambda\otimes \partial_{\lambda'})_{nn',nn'}(\rho_0^{\otimes2})_{i'i,mm'}(\rho_0^{\otimes2})_{m'm,qq'}\right]\\
& = \frac{1}{N^2(N^2-1)^2}\sum_{\lambda,\lambda'} \left(\tr \left[\partial_\lambda\otimes\partial_{\lambda'}\right]\tr \left[\partial_\lambda \otimes\partial_{\lambda'} \right]\tr \left[(\rho_0^2)^{\otimes2}\right]\right)\\
& = \frac{1}{N^2(N^2-1)^2}\sum_{\lambda,\lambda'} \left(\tr \left[\partial_\lambda\right]\tr \left[\partial_{\lambda'}\right]\tr \left[\rho_0^2\right]\right)^2
\end{aligned}
\end{equation}

\begin{equation}
\begin{aligned}
\textrm{\ding{176}}\times\textrm {\ding{181}} & \rightarrow  \frac{1}{N^2(N^2-1)^2}\sum_{\lambda,\lambda'}\tr \left[\delta_{jk}\delta_{j'k'}\delta_{il'}\delta_{i'l}\delta_{no'}\delta_{n'o}\delta_{mp}\delta_{m'p'}(\partial_\lambda\otimes \partial_{\lambda'})_{jj',kk'}(\partial_\lambda\otimes \partial_{\lambda'})_{nn',oo'}(\rho_0^{\otimes2})_{ll',mm'}(\rho_0^{\otimes2})_{pp',qq'}\right]\\
& = \frac{1}{N^2(N^2-1)^2}\sum_{\lambda,\lambda'} \tr\left[ (\partial_\lambda\otimes \partial_{\lambda'})_{jj',jj'}(\partial_\lambda\otimes \partial_{\lambda'})_{nn',n'n}(\rho_0^{\otimes2})_{i'i,mm'}(\rho_0^{\otimes2})_{mm',qq'}\right]\\
& = \frac{1}{N^2(N^2-1)^2}\sum_{\lambda,\lambda'} \left(\tr \left[\partial_\lambda \otimes\partial_{\lambda'}\right]\tr \left[\partial_\lambda \otimes\partial_{\lambda'}{\rm SWAP} \right]\tr \left[(\rho_0^2)^{\otimes2}{\rm SWAP} \right]\right)\\
& = \frac{1}{N^2(N^2-1)^2}\sum_{\lambda,\lambda'} \left(\tr \left[\partial_\lambda\partial_{\lambda'}\right]\tr \left[\partial_\lambda\right]\tr \left[\partial_{\lambda'}\right]\tr \left[\rho_0^4\right]\right)
\end{aligned}
\end{equation}

\begin{equation}
\begin{aligned}
\textrm{\ding{177}}\times\textrm {\ding{178}}=\textrm{\ding{174}}\times\textrm {\ding{181}}  \rightarrow  -\frac{1}{N(N^2-1)^2}\sum_{\lambda,\lambda'} \left(\tr \left[\partial_\lambda\right]\tr \left[\partial_{\lambda'}\right]\tr \left[\partial_\lambda \partial_{\lambda'}\right](\tr \left[\rho_0^2\right])^2\right)
\end{aligned}
\end{equation}

\begin{equation}
\begin{aligned}
\textrm{\ding{177}}\times\textrm {\ding{179}}=\textrm{\ding{175}}\times\textrm {\ding{181}}  \rightarrow   -\frac{1}{N(N^2-1)^2}\sum_{\lambda,\lambda'} \left((\tr \left[\partial_\lambda\partial_{\lambda'}\right])^2\tr \left[\rho_0^4\right]\right)
\end{aligned}
\end{equation}

\begin{equation}
\begin{aligned}
\textrm{\ding{177}}\times\textrm {\ding{180}} =\textrm{\ding{176}}\times\textrm {\ding{181}}  \rightarrow   \frac{1}{N^2(N^2-1)^2}\sum_{\lambda,\lambda'} \left(\tr \left[\partial_\lambda\partial_{\lambda'}\right]\tr \left[\partial_\lambda\right]\tr \left[\partial_{\lambda'}\right]\tr \left[\rho_0^4\right]\right)
\end{aligned}
\end{equation}

\begin{equation}
\begin{aligned}
\textrm{\ding{177}}\times\textrm {\ding{181}} & \rightarrow  \frac{1}{N^2(N^2-1)^2}\sum_{\lambda,\lambda'}\tr \left[\delta_{jk'}\delta_{j'k}\delta_{il}\delta_{i'l'}\delta_{no'}\delta_{n'o}\delta_{mp}\delta_{m'p'}(\partial_\lambda\otimes \partial_{\lambda'})_{jj',kk'}(\partial_\lambda\otimes \partial_{\lambda'})_{nn',oo'}(\rho_0^{\otimes2})_{ll',mm'}(\rho_0^{\otimes2})_{pp',qq'}\right]\\
& = \frac{1}{N^2(N^2-1)^2}\sum_{\lambda,\lambda'} \tr\left[ (\partial_\lambda\otimes \partial_{\lambda'})_{jj',j'j}(\partial_\lambda\otimes \partial_{\lambda'})_{nn',n'n}(\rho_0^{\otimes2})_{ii',mm'}(\rho_0^{\otimes2})_{mm',qq'}\right]\\
& = \frac{1}{N^2(N^2-1)^2}\sum_{\lambda,\lambda'} \left(\tr \left[\partial_\lambda \otimes\partial_{\lambda'}{\rm SWAP} \right]\tr \left[\partial_\lambda \otimes\partial_{\lambda'}{\rm SWAP} \right]\tr \left[(\rho_0^2)^{\otimes2} \right]\right)\\
& = \frac{1}{N^2(N^2-1)^2}\sum_{\lambda,\lambda'} \left(\tr \left[\partial_\lambda\partial_{\lambda'}\right]\tr \left[\rho_0^2\right]\right)^2
\end{aligned}
\end{equation}

\begin{equation}
\begin{aligned}
&\textrm{\ding{174}}\times\textrm {\ding{178}}+\textrm{\ding{174}}\times\textrm {\ding{180}}+\textrm{\ding{176}}\times\textrm {\ding{178}}+\textrm{\ding{176}}\times\textrm {\ding{180}} \rightarrow\\
& \frac{1}{(N^2-1)^2}\sum_{\lambda,\lambda'} \left(\left(\tr \left[\partial_\lambda\right]\tr \left[ \partial_{\lambda'}\right]\tr\left[\rho_0^2\right]\right)^2-\frac{2}{N} \left((\tr \left[\partial_\lambda\right])^2(\tr \left[\partial_{\lambda'}\right])^2\tr \left[\rho_0^4\right]\right)+\frac{1}{N^2} \left(\tr \left[\partial_\lambda\right]\tr \left[\partial_{\lambda'}\right]\tr \left[\rho_0^2\right]\right)^2\right)\\
=& \frac{1}{(N^2-1)^2}\sum_{\lambda,\lambda'}(\tr \left[\partial_\lambda\right]\tr \left[\partial_{\lambda'}\right])^2\tr \left[\left(\rho_0^{\otimes2}-\frac{1}{N}\rho_0^{\otimes2}{\rm SWAP}\right)^2\right].
\end{aligned}
\end{equation}

\begin{equation}
\begin{aligned}
&\textrm{\ding{174}}\times\textrm {\ding{179}}+\textrm{\ding{175}}\times\textrm {\ding{178}}+\textrm{\ding{174}}\times\textrm {\ding{181}}+\textrm{\ding{177}}\times\textrm {\ding{178}} +\textrm{\ding{175}}\times\textrm {\ding{180}}+\textrm{\ding{176}}\times\textrm {\ding{179}}+\textrm{\ding{176}}\times\textrm {\ding{181}}+\textrm{\ding{177}}\times\textrm {\ding{180}} \rightarrow\\
&  \frac{2}{(N^2-1)^2}\sum_{\lambda,\lambda'} \left(\tr \left[\partial_\lambda\right]\tr \left[\partial_{\lambda'}\right]\tr \left[\partial_\lambda \partial_{\lambda'}\right]\tr \left[\rho_0^4\right]\right)-\frac{2}{N(N^2-1)^2}\sum_{\lambda,\lambda'} \left(\tr \left[\partial_\lambda\right]\tr \left[\partial_{\lambda'}\right]\tr \left[\partial_\lambda \partial_{\lambda'}\right](\tr \left[\rho_0^2\right])^2\right)\\
&  -\frac{2}{N(N^2-1)^2}\sum_{\lambda,\lambda'} \left(\tr \left[\partial_\lambda \partial_{\lambda'}\right]\tr \left[\partial_\lambda\right]\tr \left[\partial_{\lambda'}\right](\tr \left[\rho_0^2\right])^2\right) +\frac{2}{N^2(N^2-1)^2}\sum_{\lambda,\lambda'} \left(\tr \left[\partial_\lambda\partial_{\lambda'}\right]\tr \left[\partial_\lambda\right]\tr \left[\partial_{\lambda'}\right]\tr \left[\rho_0^4\right]\right)\\
=& \frac{2}{(N^2-1)^2}\sum_{\lambda,\lambda'} \tr \left[\partial_\lambda\partial_{\lambda'}\right]\tr \left[\partial_\lambda\right]\tr \left[\partial_{\lambda'}\right]\left(\tr\left[ \rho_0^4\right]-\frac{2}{N}\left(\tr\left[ \rho_0^2\right]\right)^2+\frac{1}{N^2}\tr\left[ \rho_0^4\right]\right)\\
=& \frac{2}{(N^2-1)^2}\sum_{\lambda,\lambda'} \tr \left[\partial_\lambda\partial_{\lambda'}\right]\tr \left[\partial_\lambda\right]\tr \left[\partial_{\lambda'}\right]\tr \left[\left(\rho_0^{\otimes2}-\frac{1}{N}\rho_0^{\otimes2}{\rm SWAP}\right)\left(\rho_0^{\otimes2}{\rm SWAP}-\frac{1}{N}\rho_0^{\otimes2}\right)\right].
\end{aligned}
\end{equation}

\begin{equation}
\begin{aligned}
&\textrm{\ding{175}}\times\textrm {\ding{179}}+\textrm{\ding{175}}\times\textrm {\ding{181}}+\textrm{\ding{177}}\times\textrm {\ding{179}}+\textrm{\ding{177}}\times\textrm {\ding{181}} \rightarrow\\
& \frac{1}{(N^2-1)^2}\sum_{\lambda,\lambda'} \left(\left(\tr \left[\partial_\lambda \partial_{\lambda'}\right]\tr \left[\rho_0^2\right]\right)^2 -\frac{2}{N}\left((\tr \left[\partial_\lambda\partial_{\lambda'}\right])^2\tr \left[\rho_0^4\right]\right)+ \frac{1}{N^2} \left(\tr \left[\partial_\lambda\partial_{\lambda'}\right]\tr \left[\rho_0^2\right]\right)^2\right)\\
=& \frac{1}{(N^2-1)^2}\sum_{\lambda,\lambda'} (\tr \left[\partial_\lambda\partial_{\lambda'}\right])^2\tr \left[\left(\rho_0^{\otimes2}-\frac{1}{N}\rho_0^{\otimes2}{\rm SWAP}\right)^2\right]\\
=& \frac{1}{(N^2-1)^2}\sum_{\lambda,\lambda'} (\tr \left[\partial_\lambda\partial_{\lambda'}\right])^2\tr \left[\left(\rho_0^{\otimes2}{\rm SWAP}-\frac{1}{N}\rho_0^{\otimes2}\right)^2\right]
\end{aligned}
\end{equation}
Summing over all the terms, then we have the following variance term 
\begin{equation}\label{Var_K_Express}
\begin{aligned}
{\rm Var}[K_{\vec{x}, \vec{x}'}] =&\frac{1}{(N^2-1)^2}\sum_{\lambda,\lambda'}\tr\left[\left[\left(\rho_0^{\otimes2}-\frac{1}{N}\rho_0^{\otimes2}{\rm SWAP}\right)\tr[\partial_\lambda]\tr[\partial_{\lambda'}]+\left(\rho_0^{\otimes2}{\rm SWAP}-\frac{1}{N}\rho_0^{\otimes2}\right)\tr[\partial_\lambda\partial_{\lambda'}]\right]^2\right]\\
&+\mathcal{O}((\mathcal{M}^{[1]}_{\mathbb{V}})^2,\mathcal{M}^{[2]}_{\mathbb{V}}).
\end{aligned}
\end{equation}

If we consider the structure of quantum circuit defined  in Eq.~(\ref{U_parametrize}), we have 
\begin{equation}
\begin{aligned}
\partial_\lambda\equiv&\partial_{\theta_\lambda} (U^\dag(\theta)\hat{O}U(\theta))\\
=&(\partial_{\theta_\lambda} U^\dag(\theta))\hat{O}U(\theta)+ U^\dag(\theta)\hat{O}(\partial_{\theta_\lambda}U(\theta))\\
=&U_{\lambda,+}^\dag(-iW_\lambda)U_{\lambda,-}^\dag\hat{O}U_{\lambda,-}U_{\lambda,+}+U_{\lambda,+}^\dag U_{\lambda,-}^\dag\hat{O}U_{\lambda,-}(iW_\lambda)U_{\lambda,+}\\
=& i U_{\lambda,+}^\dag \left[ U_{\lambda,-}^\dag\hat{O}U_{\lambda,-}, \, W_\lambda \right] U_{\lambda,+},
\end{aligned}
\end{equation}
then one has $\tr[\partial_\lambda]=0$. Thus the variance of $K_{\bf x,x'}$ in Eq.~(\ref{Var_K_Express}) has the form
\begin{equation}\label{Var_K_Express_Simplify}
\begin{aligned}
{\rm Var}[K_{\vec{x}, \vec{x}'}] =\frac{1}{(N^2-1)^2}\sum_{\lambda,\lambda'}\tr\left[\left(\rho_0^{\otimes2}{\rm SWAP}-\frac{1}{N}\rho_0^{\otimes2}\right)^2\right]\left(\tr[\partial_\lambda\partial_{\lambda'}]\right)^2-\mathcal{O}((\mathcal{M}^{[1]}_{\mathbb{V}})^2,\mathcal{M}^{[2]}_{\mathbb{V}}).
\end{aligned}
\end{equation}
Using the Cauchy-Schwartz inequality for Hermitian operators $A$ and $B$
\begin{equation}
\left(\tr[AB]\right)^2\leq \tr[A^2]\tr[B^2],
\end{equation}
one has 
\begin{equation}
\left(\tr[\partial_\lambda\partial_{\lambda'}]\right)^2\leq \tr[\partial_\lambda^2]\tr[\partial_{\lambda'}^2],
\end{equation}
\begin{equation}\label{Cauchy_Lambda}
\begin{aligned}
\tr[\partial_\lambda^2] = & -\tr\left[ U_{\lambda,+}^\dag \left[ U_{\lambda,-}^\dag\hat{O}U_{\lambda,-}, \, W_\lambda \right] U_{\lambda,+}U_{\lambda,+}^\dag \left[ U_{\lambda,-}^\dag\hat{O}U_{\lambda,-}, \, W_\lambda \right] U_{\lambda,+}\right]\\
= & -\tr\left[  \left[ U_{\lambda,-}^\dag\hat{O}U_{\lambda,-}, \, W_\lambda \right]  \left[ U_{\lambda,-}^\dag\hat{O}U_{\lambda,-}, \, W_\lambda \right] \right]\\
= & -\tr\left[  \left( U_{\lambda,-}^\dag\hat{O}U_{\lambda,-}W_\lambda-W_\lambda U_{\lambda,-}^\dag\hat{O}U_{\lambda,-} \right)  \left( U_{\lambda,-}^\dag\hat{O}U_{\lambda,-}W_\lambda-W_\lambda U_{\lambda,-}^\dag\hat{O}U_{\lambda,-} \right)  \right]\\
=& -2\tr\left[ U_{\lambda,-}^\dag\hat{O}U_{\lambda,-}W_\lambda U_{\lambda,-}^\dag\hat{O}U_{\lambda,-}W_\lambda-U_{\lambda,-}^\dag\hat{O}U_{\lambda,-}W_\lambda W_\lambda U_{\lambda,-}^\dag\hat{O}U_{\lambda,-}\right]\\
=& -2\tr\left[ U_{\lambda,-}^\dag\hat{O}U_{\lambda,-}W_\lambda U_{\lambda,-}^\dag\hat{O}U_{\lambda,-}W_\lambda-\hat{O}^2\right]\\
=& 2\tr\left[ \hat{O}^2-\left(U_{\lambda,-}^\dag\hat{O}U_{\lambda,-}W_\lambda\right)^2\right]\\
\leq & 2\tr\left[ \hat{O}^2\right].
\end{aligned}
\end{equation}
Thus the inequality can be reexpressed as 
\begin{equation}\label{Supp_Variance}
\begin{aligned}
{\rm Var}[K_{\vec{x}, \vec{x}'}] \leq & \frac{1}{(N^2-1)^2}\sum_{\lambda,\lambda'}\tr\left[\left(\rho_0^{\otimes2}{\rm SWAP}-\frac{1}{N}\rho_0^{\otimes2}\right)^2\right]\left(2\tr\left[ \hat{O}^2\right]\right)^2-\mathcal{O}((\mathcal{M}^{[1]}_{\mathbb{V}})^2,\mathcal{M}^{[2]}_{\mathbb{V}})\\
& < \frac{\Lambda^2(2\tr[ \hat{O}^2])^2}{(2^{2n}-1)^2}-\mathcal{O}((\mathcal{M}^{[1]}_{\mathbb{V}})^2,\mathcal{M}^{[2]}_{\mathbb{V}}),
\end{aligned}
\end{equation}
where $\Lambda$ denotes the total number of variational parameters in the quantum machine learning model, and the dimension of Hilbert space $N=2^n$. 
Substituting the results of the expectation and variance of  the QNTK elements in Eq.~(\ref{Supp_Expectation_Global}) and Ineq.~(\ref{Supp_Variance}) into the Chebyschev's inequality, one obtains
\begin{equation}
\begin{aligned}
{\rm Pr}_{\vec{x}, \vec{x}'}\left(|K(\vec{x}, \vec{x}')+\mathcal{O}(\mathcal{M}^{[1]}_{\mathbb{V}})|\geq \epsilon\right)
\leq  \frac{\epsilon^{-2}\Lambda^2}{(2^{2n}-1)^2}\left(2 {{\rm tr}}[ \hat{O}^2]\right)^2-\epsilon^{-2}\mathcal{O}((\mathcal{M}^{[1]}_{\mathbb{V}})^2,\mathcal{M}^{[2]}_{\mathbb{V}}).
\end{aligned}
\end{equation}
This completes our proof.

 \section{Proof of Theorem 2: Local expressibility induce concentration of QNTK}
Here we consider the case of quantum encoding through the  tensor-product of $n$ local unitary operators, i.e., the local quantum encoding, 
\begin{equation}
\mathbb{W}_{\vec{x}}: {\vec{x}} \longmapsto W({\vec{x}}) =\bigotimes_{i=1}^{n} W_i (\vec{x}), \quad W_i (\vec{x}) \in {\rm SU} (2).
\end{equation}
 Thus,  the classical data samples are encoded into the $n$-qubit  local tensor product states $|\psi(\vec{x})\rangle = \bigotimes_{i=1}^{n} W_i (\vec{x})|0\rangle^{\otimes n}$. 
 
 In the case of local quantum encoding, we have the following theorem. 
 
 \begin{theorem}\label{Supp_theorem:local} 
In the case of local quantum encoding,  for  arbitrary two input samples $\vec{x}$ and $\vec{x}'$ drawn from the same distribution, the QNTK  $K(\vec{x}, \vec{x}')$ obeys
\begin{equation}\label{Supp_QNTK_Distribution_Local}
\begin{aligned}
{\rm Pr}_{\vec{x}, \vec{x}'}\left(|K(\vec{x}, \vec{x}')+\mathcal{O}(\mathcal{M}^{[1]}_{\mathbb{W}})|\geq \epsilon\right)
\leq  \frac{\epsilon^{-2}\Lambda^2}{2^{2n-2}}\left( {{\rm tr}}[ \hat{O}^2]\right)^2-\epsilon^{-2}\mathcal{O}((\mathcal{M}^{[1]}_{\mathbb{W}})^2,\mathcal{M}^{[2]}_{\mathbb{W}}).
\end{aligned}
\end{equation}
for arbitrary small positive constant $\epsilon$, where $n$ denotes the number of qubits, $\Lambda$ denotes the number of trainable parameters, and $\mathcal{M}^{[1,2]}_{\mathbb{W}}$ means the measure of the $t$-moment local expressibility of unitary ensemble $\mathbb{W}$.
\end{theorem}

 Now we provide the detailed proof of Theorem 2. Here we also adopt the 
 Chebyschev's inequality, which has been mentioned in Ineq.~(\ref{Chebyschev_Inequality}). In the following subsections,  we mainly focus on deriving the analytical formula of  the expectation value $\mathbb{E}[K_{\vec{x}, \vec{x}'}]$ and variance ${\rm Var}[K_{\vec{x}, \vec{x}'}]$  of the quantum tangent kernel matrix elements with respect to the locally  encoded states.
 
 \subsection{Expectation of the QNTK over the  locally encoded  space}

We first calculate the expectation value of the QNTK with respect to the locally encoded states. Here $N$ denotes the dimension of $n$-qubit Hilbert space, $N\equiv2^n$. For brevity, here we denote the local tensor product unitary operator by 
\begin{equation}
W_{jk} =W_{(j_1,j_2,\cdots j_n)(k_1,k_2,\cdots k_n)}.
\end{equation} 
The 1-moment integral over Haar ensemble of the local tensor product unitary group takes the form
\begin{equation}\label{1_design_local}
\begin{aligned}
M_1(dU_{{\rm Haar}^{\otimes n}})&=\int_{{\rm Haar}^{\otimes n}}dU\;U_{j,k}\bar{U}_{l,m} \\
&=\int_{{\rm Haar}^{\otimes n}}dU\;U_{j_1j_2\cdots j_n,k_1k_2\cdots k_n}\bar{U}_{l_1l_2\cdots l_n,m_1m_2\cdots m_n} \\
& = \frac{1}{2^n}\prod_{i=1}^n\delta_{j_il_i}\delta_{k_in_i}\\
& =  \frac{1}{2^n} \delta_{(j_1,j_2,\cdots j_n)(l_1,l_2,\cdots l_n)} \delta_{(k_1,k_2,\cdots k_n)(m_1,m_2,\cdots m_n)}\\
& = \frac{1}{2^n} \delta_{jl} \delta_{km}.
\end{aligned}
\end{equation}

Then the expectation value of $K_{\vec{x},  \vec{x}'}$ over the local encoding ensemble $\mathbb{W}$ takes the similar form  of Eq.~(\ref{Supp_Expectation_Global})
\begin{equation}\label{Supp_Expectation_Local}
\begin{aligned}
\mathbb{E}[K_{\vec{x},  \vec{x}'}]&=\sum_{\lambda} \int_{\mathbb{W}} dW(\vec{x})dW(\vec{x}')\tr \left[ W^\dag(\vec{x})\partial_{\theta_\lambda} \hat{O}(\theta)W(\vec{x})\rho_0 W^\dag(\vec{x}')\partial_{\theta_\lambda} \hat{O}(\theta)W(\vec{x}')\rho_0\right]\\
&=\sum_{\lambda} \int_{\mathbb{W}} dW(\vec{x})dW(\vec{x}')\tr \left[ W^\dag(\vec{x})_{ij}(\partial_{\theta_\lambda} \hat{O}(\theta))_{jk}W(\vec{x})_{kl}(\rho_0)_{lm} W^\dag(\vec{x}')_{mn}(\partial_{\theta_\lambda} \hat{O}(\theta))_{no}W(\vec{x}')_{op}(\rho_0)_{pq}\right]\\
&=\sum_{\lambda} \tr\left[\left(\frac{1}{N}\delta_{il}\delta_{jk}-(\mathcal{A}^{[1]}_{\mathbb{W}})_{il,jk}\right)(\partial_{\theta_\lambda} \hat{O}(\theta))_{jk}(\rho_0)_{lm}\left(\frac{1}{N}\delta_{mp}\delta_{no}-(\mathcal{A}^{[1]}_{\mathbb{W}})_{mp,no}\right)(\partial_{\theta_\lambda} \hat{O}(\theta))_{no}(\rho_0)_{pq}\right]\\
&=\sum_{\lambda}\left[\left(\frac{1}{N}\delta_{il}\delta_{jk}-(\mathcal{A}^{[1]}_{\mathbb{W}})_{il,jk}\right)(\partial_{\theta_\lambda} \hat{O}(\theta))_{jk}(\rho_0)_{lm}\left(\frac{1}{N}\delta_{mp}\delta_{no}-(\mathcal{A}^{[1]}_{\mathbb{W}})_{mp,no}\right)(\partial_{\theta_\lambda} \hat{O}(\theta))_{no}(\rho_0)_{pi}\right]\\
&=\sum_{\lambda} \left[N^{-2}\delta_{il}\delta_{jk}\delta_{mp}\delta_{no}-N^{-1}\delta_{il}\delta_{jk}(\mathcal{A}^{[1]}_{\mathbb{W}})_{mp,no}-N^{-1}\delta_{mp}\delta_{no}(\mathcal{A}^{[1]}_{\mathbb{W}})_{il,jk}+(\mathcal{A}^{[1]}_{\mathbb{W}})_{il,jk}(\mathcal{A}^{[1]}_{\mathbb{W}})_{mp,no}\right]\\
& \qquad \quad \times (\partial_{\theta_\lambda} \hat{O}(\theta))_{jk}(\rho_0)_{lm}(\partial_{\theta_\lambda} \hat{O}(\theta))_{no}(\rho_0)_{pi}\\
&= \frac{1}{N^2} \sum_\lambda (\partial_{\theta_\lambda} \hat{O}(\theta))_{jj}(\partial_{\theta_\lambda} \hat{O}(\theta))_{nn}(\rho_0)_{ip}(\rho_0)_{pi} -\mathcal{O}(\mathcal{M}^{[1]}_{\mathbb{W}})\\
&=  \frac{1}{N^2} \sum_\lambda \left(\tr\left[\partial_{\theta_\lambda} \hat{O}(\theta)\right]\right)^2-\mathcal{O}(\mathcal{M}^{[1]}_{\mathbb{W}})\\
&=  \frac{1}{2^{2n}} \sum_\lambda \left(\tr\left[\partial_{\theta_\lambda} \hat{O}(\theta)\right]\right)^2-\mathcal{O}(\mathcal{M}^{[1]}_{\mathbb{W}}),
\end{aligned}
\end{equation}
where $\mathcal{M}^{[1]}_{\mathbb{W}}$ measures the 1-moment  expressibility of the local tensor product, see Eq.~(\ref{Supp_Measure_Local}).

Thus for sufficient 1-moment expressibility with  $\mathcal{M}^{[1]}_{\mathbb{W}}=0$, under the parametrization of variational circuit in Eq.~(\ref{U_parametrize}), one has the expectation value of $K_{\vec{x}, \vec{x}'}$
\begin{equation}
\mathbb{E}[K_{\vec{x},  \vec{x}'}]= -\mathcal{O}(\mathcal{M}^{[1]}_{\mathbb{W}}) = 0,
\end{equation}
for maximal expressibility of the unitary ensemble $\mathbb{W}$.

\subsection{Variance of the QNTK over the  locally encoded  space}

Here we calculate the variance of the QNTK over the input data space under local quantum encoding. Without loss of generality, we suppose that $\rho_0 = (|0\rangle\langle0|)^{\otimes n}$.  
\begin{equation}
\begin{aligned}
{\rm Var}[K_{\vec{x},  \vec{x}'}] =& \mathbb{E}[K_{\vec{x},  \vec{x}'}^2] - \left(\mathbb{E}[K_{\vec{x},  \vec{x}'}]\right)^2= \mathbb{E}[K_{\vec{x},  \vec{x}'}^2] -\mathcal{O}((\mathcal{M}^{[1]}_{\mathbb{W}})^2)\\
=& \sum_{\lambda,\lambda'} \int_{\mathbb{W}} dW(\vec{x})dW(\vec{x}')\tr\left[W^\dag(\vec{x})\partial_{\theta_\lambda} \hat{O}(\theta)W(\vec{x})\rho_0 W^\dag(\vec{x}')\partial_{\theta_\lambda} \hat{O}(\theta)W(\vec{x}')\rho_0\right]\\
&\qquad \cdot \tr\left[W^\dag(\vec{x})\partial_{\theta_{\lambda'}} \hat{O}(\theta)W(\vec{x})\rho_0 W^\dag(\vec{x}')\partial_{\theta_{\lambda'}} \hat{O}(\theta)W(\vec{x}')\rho_0\right]-\mathcal{O}((\mathcal{M}^{[1]}_{\mathbb{W}})^2) \\
=& \sum_{\lambda,\lambda'}\tr \left[ {\int_{\mathbb{W}} dW(\vec{x}) (W^\dag(\vec{x}))^{\otimes2} (\partial_\lambda\otimes \partial_{\lambda'})(W(\vec{x}))^{\otimes2}}\rho_0^{\otimes2} {\int_{\mathbb{W}} dW(\vec{x}') (W^\dag(\vec{x}'))^{\otimes2} (\partial_\lambda\otimes \partial_{\lambda'})(W(\vec{x}'))^{\otimes2}}\rho_0^{\otimes2}\right]\\
&\qquad -\mathcal{O}\left(\left(\mathcal{M}^{[1]}_{\mathbb{W}}\right)^2\right) \\
=& \sum_{\lambda,\lambda'}\langle 0|^{\otimes 2n} {\int_{\mathbb{W}} dW(\vec{x}) (W^\dag(\vec{x}))^{\otimes2} (\partial_\lambda\otimes \partial_{\lambda'})(W(\vec{x}))^{\otimes2}}|0\rangle^{\otimes 2n} \\
& \qquad \cdot \langle 0|^{\otimes 2n}{\int_{\mathbb{W}} dW(\vec{x}') (W^\dag(\vec{x}'))^{\otimes2} (\partial_\lambda\otimes \partial_{\lambda'})(W(\vec{x}'))^{\otimes2}}|0\rangle^{\otimes 2n}-\mathcal{O}\left(\left(\mathcal{M}^{[1]}_{\mathbb{W}}\right)^2\right)
 \nonumber 
\end{aligned}
\end{equation}

\begin{equation} \label{Supp_Var_Local}
\begin{aligned}
=& \sum_{\lambda,\lambda'}\langle 0|^{\otimes 2n} {\int_{{\rm Haar}^{\otimes n}} dW(\vec{x}) (W^\dag(\vec{x}))^{\otimes2} (\partial_\lambda\otimes \partial_{\lambda'})(W(\vec{x}))^{\otimes2}}|0\rangle^{\otimes 2n}\\
& \qquad \cdot \langle 0|^{\otimes 2n}{\int_{{\rm Haar}^{\otimes n}} dW(\vec{x}') (W^\dag(\vec{x}'))^{\otimes2} (\partial_\lambda\otimes \partial_{\lambda'})(W(\vec{x}'))^{\otimes2}}|0\rangle^{\otimes 2n}-\mathcal{O}\left(\mathcal{M}^{[2]}_{\mathbb{W}}\right)-\mathcal{O}\left(\left(\mathcal{M}^{[1]}_{\mathbb{W}}\right)^2\right).
\end{aligned}
\end{equation}
The term $\partial_{\theta_\lambda} \hat{O}(\theta)$ is abbreviated as $\partial_\lambda$ for latter convenience. 

Through straightforward calculations, one can  show that
\begin{equation}\label{2_moment_Haar_Local}
\begin{aligned}
&\langle 0|^{\otimes 2n}{\int_{{\rm Haar}^{\otimes n}} dW(\vec{x}') (W^\dag(\vec{x}'))^{\otimes2} (\partial_\lambda\otimes \partial_{\lambda'})(W(\vec{x}'))^{\otimes2}}|0\rangle^{\otimes 2n}\\
=& \frac{1}{(2^2-1)^n}\left(1-\frac{1}{2}\right)^n\tr\left[\bigotimes_{i=1}^n\left(\mathbb{I}+{\rm SWAP}_{i,i}\right)\partial_\lambda\otimes \partial_{\lambda'}\right]\\
=& \frac{1}{6^n}\tr\left[\bigotimes_{i=1}^n\left(\mathbb{I}+{\rm SWAP}_{i,i}\right)\partial_\lambda\otimes \partial_{\lambda'}\right],
\end{aligned}
\end{equation}
where the swap operator ${\rm SWAP}_{i,i}$ acts on the $i$-th qubit subsystem of the operator $\partial_\lambda\otimes \partial_{\lambda'}$. Substituting Eq.~(\ref{2_moment_Haar_Local}) into the variance in Eq.~(\ref{Supp_Var_Local}), together with Ineq.~(\ref{Cauchy_Lambda}), one then obtains 
\begin{equation}\label{Supp_Var_Local_Eq}
\begin{aligned}
{\rm Var}[K_{\vec{x},  \vec{x}'}] =& \mathbb{E}[K_{\vec{x},  \vec{x}'}^2] - \left(\mathbb{E}[K_{\vec{x},  \vec{x}'}]\right)^2= \mathbb{E}[K_{\vec{x},  \vec{x}'}^2] -\mathcal{O}((\mathcal{M}^{[1]}_{\mathbb{W}})^2)\\
=& \sum_{\lambda,\lambda'}\langle 0|^{\otimes 2n} {\int_{{\rm Haar}^{\otimes n}} dW(\vec{x}) (W^\dag(\vec{x}))^{\otimes2} (\partial_\lambda\otimes \partial_{\lambda'})(W(\vec{x}))^{\otimes2}}|0\rangle^{\otimes 2n}\\
& \qquad \cdot \langle 0|^{\otimes 2n}{\int_{{\rm Haar}^{\otimes n}} dW(\vec{x}') (W^\dag(\vec{x}'))^{\otimes2} (\partial_\lambda\otimes \partial_{\lambda'})(W(\vec{x}'))^{\otimes2}}|0\rangle^{\otimes 2n}-\mathcal{O}\left(\mathcal{M}^{[2]}_{\mathbb{W}}\right)-\mathcal{O}\left(\left(\mathcal{M}^{[1]}_{\mathbb{W}}\right)^2\right)\\ 
 = &  \frac{1}{6^{2n}} \sum_{\lambda,\lambda'} \left(\tr\left[\bigotimes_{i=1}^n\left(\mathbb{I}+{\rm SWAP}_{i,i}\right)\partial_\lambda\otimes \partial_{\lambda'}\right]\right)^2-\mathcal{O}\left(\left(\mathcal{M}^{[1]}_{\mathbb{W}}\right)^2,\mathcal{M}^{[2]}_{\mathbb{W}}\right).
\end{aligned}
\end{equation}
Using the matrix trace inequality \cite{Renaud1998Some}
\begin{equation}
\left|\tr[W A^\dag B]\right|^2\leq\tr[W A^\dag A]\tr[W B^\dag B],
\end{equation}
where $W$ is positive semi-definite with $\tr(W)=1$, one obtains that
\begin{equation}
\begin{aligned}
&\left(\tr\left[\bigotimes_{i=1}^n\left(\mathbb{I}+{\rm SWAP}_{i,i}\right)\partial_\lambda\otimes \partial_{\lambda'}\right]\right)^2 \\
 \leq& \tr\left[\bigotimes_{i=1}^n\left(\mathbb{I}+{\rm SWAP}_{i,i}\right)\partial_\lambda^2\otimes \mathbb{I}_{2^n}\right]\tr\left[\bigotimes_{i=1}^n\left(\mathbb{I}+{\rm SWAP}_{i,i}\right)\mathbb{I}_{2^n}\otimes \partial_{\lambda'}^2\right]\\
 =& (2+1)^n\tr\left[\partial_\lambda^2\right]\cdot (2+1)^n\tr\left[\partial_{\lambda'}^2\right]\\
 =& 3^{2n}\tr\left[\partial_\lambda^2\right]\tr\left[\partial_{\lambda'}^2\right].
\end{aligned}
\end{equation}

Then the variance in Eq.~(\ref{Supp_Var_Local_Eq}) satisfies
\begin{equation}\label{Supp_Var_Local_Ineq}
\begin{aligned}
{\rm Var}[K_{\vec{x},  \vec{x}'}] =& \mathbb{E}[K_{\vec{x},  \vec{x}'}^2] - \left(\mathbb{E}[K_{\vec{x},  \vec{x}'}]\right)^2= \mathbb{E}[K_{\vec{x},  \vec{x}'}^2] -\mathcal{O}((\mathcal{M}^{[1]}_{\mathbb{W}})^2)\\
=& \sum_{\lambda,\lambda'}\langle 0|^{\otimes 2n} {\int_{{\rm Haar}^{\otimes n}} dW(\vec{x}) (W^\dag(\vec{x}))^{\otimes2} (\partial_\lambda\otimes \partial_{\lambda'})(W(\vec{x}))^{\otimes2}}|0\rangle^{\otimes 2n}\\
& \qquad \cdot \langle 0|^{\otimes 2n}{\int_{{\rm Haar}^{\otimes n}} dW(\vec{x}') (W^\dag(\vec{x}'))^{\otimes2} (\partial_\lambda\otimes \partial_{\lambda'})(W(\vec{x}'))^{\otimes2}}|0\rangle^{\otimes 2n}-\mathcal{O}\left(\mathcal{M}^{[2]}_{\mathbb{W}}\right)-\mathcal{O}\left(\left(\mathcal{M}^{[1]}_{\mathbb{W}}\right)^2\right)\\ 
 = &  \frac{1}{6^{2n}} \sum_{\lambda,\lambda'} \left(\tr\left[\bigotimes_{i=1}^n\left(\mathbb{I}+{\rm SWAP}_{i,i}\right)\partial_\lambda\otimes \partial_{\lambda'}\right]\right)^2-\mathcal{O}\left(\left(\mathcal{M}^{[1]}_{\mathbb{W}}\right)^2,\mathcal{M}^{[2]}_{\mathbb{W}}\right)\\
 \leq & \frac{1}{6^{2n}} \cdot 3^{2n}\sum_{\lambda,\lambda'} \tr\left[\partial_\lambda^2\otimes \partial_{\lambda'}^2\right]-\mathcal{O}\left(\left(\mathcal{M}^{[1]}_{\mathbb{W}}\right)^2,\mathcal{M}^{[2]}_{\mathbb{W}}\right)\\
 = & \frac{1}{2^{2n}} \sum_{\lambda,\lambda'}\tr\left[\partial_\lambda^2\right]\tr\left[\partial_{\lambda'}^2\right]-\mathcal{O}\left(\left(\mathcal{M}^{[1]}_{\mathbb{W}}\right)^2,\mathcal{M}^{[2]}_{\mathbb{W}}\right)\\
 \leq &  \frac{\Lambda^2}{2^{2n-2}}\left(\tr\left[ \hat{O}^2\right]\right)^2-\mathcal{O}\left(\left(\mathcal{M}^{[1]}_{\mathbb{W}}\right)^2,\mathcal{M}^{[2]}_{\mathbb{W}}\right),
\end{aligned}
\end{equation}
where $\Lambda$ denotes the total number of variational parameters in the quantum machine learning model. 
Substituting the results of the expectation and variance of  the QNTK elements in Eq.~(\ref{Supp_Expectation_Local}) and Ineq.~(\ref{Supp_Var_Local_Ineq}) into the Chebyschev's inequality, one obtains
\begin{equation}
\begin{aligned}
{\rm Pr}_{\vec{x}, \vec{x}'}\left(|K(\vec{x}, \vec{x}')+\mathcal{O}(\mathcal{M}^{[1]}_{\mathbb{W}})|\geq \epsilon\right)
\leq  \frac{\epsilon^{-2}\Lambda^2}{2^{2n-2}}\left({{\rm tr}}[ \hat{O}^2]\right)^2-\epsilon^{-2}\mathcal{O}((\mathcal{M}^{[1]}_{\mathbb{W}})^2,\mathcal{M}^{[2]}_{\mathbb{W}}).
\end{aligned}
\end{equation}
This completes our proof.

\begin{figure}
\centering
\includegraphics[width=0.5\linewidth]{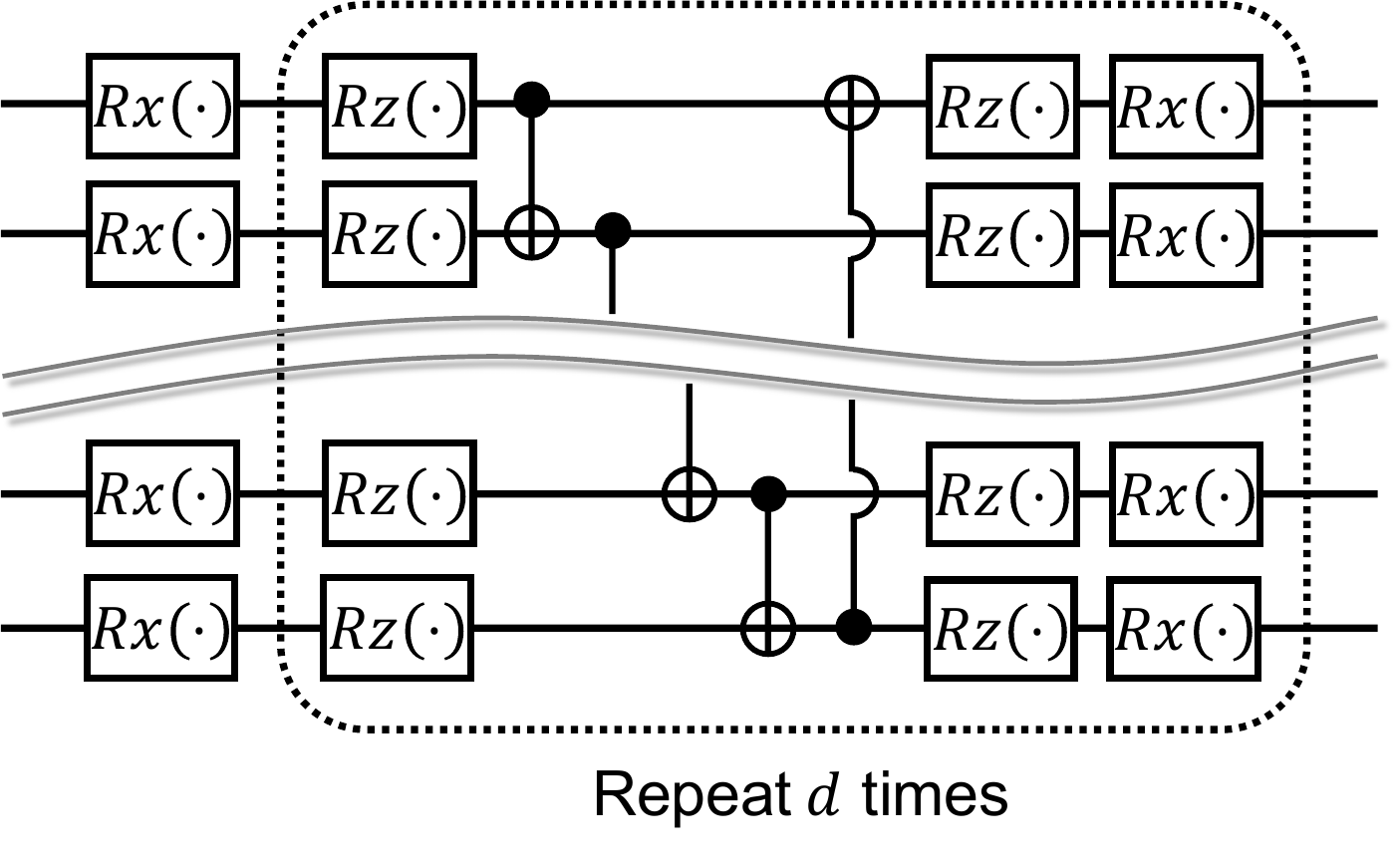} 
\caption{ {Variational circuit structure for our numerical simulations. } $R_x$ ($R_z$) denotes the single-qubit rotational gate along the $x$-direction  ($z$-direction), and each rotational gate possesses one independent variational parameter. The dotted-line block consists of single-qubit rotational gates and two-qubit entangling gates (CNOT). To achieve arbitrary circuit depth, one can repeat the block for $d$ times, thus the depth of the quantum circuit is proportional to the block number $d$.}\label{Supp_Var_Circ}
\end{figure}

\section{Numerical Details}

We calculate the expectation value and variance of the quantum tangent kernel matrix elements based on the open source package Yao.jl \cite{Luo2020Yaojlextensible} in Julia language. 

\subsection{Global quantum encoding}
In our setting of encoding layers, we map the classical input sample $\vec{x}=(x_1,x_2,\cdots ,x_N)$ into quantum states through the quantum amplitude encoding:
$$\vec{x}\rightarrow|\psi_{\vec{x}}\rangle=\sum_{j=1}^N x_j|j\rangle,$$ where $\{|j\rangle\}$ is a set of orthonormal quantum states. 

We choose the default setting of the variational layers defined in Yao.jl,  see Fig.~\ref{Supp_Var_Circ} for details. For arbitrary number of qubits, the dotted-line block in Fig.~\ref{Supp_Var_Circ} consists of single-qubit rotational gates and two-qubit control-NOT gates, and each local rotational gate hosts one independent variational parameter. One can then improve the number of variational parameters by repeatedly applying the dotted-line block on the quantum circuit.  We initially randomize the variational parameters in the circuit. 

In calculating the expectation and variance of QNTK values, we mainly focus on two types of mean square loss functions, including the global and local ones.

\subsubsection{The global loss function}
In the case of global loss function, we choose the loss function as the following formula,
\begin{equation}\label{Supp_Global}
\mathcal{L}(\vec{\theta}) = \frac{1}{2}\sum_i\left(\langle{\vec{x}}_i |U^\dag(\vec{\theta})\left(|0\rangle\langle 0|\right)^{\otimes n}U(\vec{\theta})|{\vec{x}}_i\rangle-y_i\right)^2,
\end{equation}
where $\vec{\theta}$ represents the set of variational  parameters,  ${\vec{x}}_i$ and $y_i$ represent the feature and the label of the $i$-th data sample, respectively. Here we chose the global observable $\hat{O}$ to be a pure state density matrix, namely, $\hat{O}\equiv(|0\rangle\langle 0|)^{\otimes n}$.

With the above global loss function, we then obtain the corresponding QNTK matrix element for arbitrary two input samples $\vec{x}$ and $\vec{x}'$
\begin{equation}
K_{\vec{x},\vec{x}'}=K(\vec{x},\vec{x}')=\sum_{\lambda}\frac{\partial \left|\langle{\vec{x}} |U^\dag(\vec{\theta})|0\rangle^{\otimes n}\right|^2}{\partial \theta_{\lambda}} \frac{\partial\left|\langle{\vec{x}'} |U^\dag(\vec{\theta})|0\rangle^{\otimes n}\right|^2}{\partial \theta_{\lambda}}.
\end{equation}

Then we numerically calculate the expectation value and variance of $K_{\vec{x},\vec{x}'}$ over the input data set $\mathcal{X}=\{\vec{x}\}$. We assume that the quantum encoding circuit is highly expressive, and randomly choose the classical input data from the complex space $(\mathbb{C}^{2})^{\otimes n}$. For convenience, in practical calculations, we randomly choose the quantum states in $n$-qubit Hilbert space. 

We compute the expectation value and variance of QNTK in various cases involving different numbers of qubits and circuit depths. The qubit number ranges from four to eleven, and the number of block layers $d$ takes the cases of $\{5, 10, 20, 50, 100, 200\}$.  The full numerical results have been presented in the main text.  We notice that both the expectation value and variance of QNTK  decay exponentially to zero with respect to the system size $n$, which are consistent with our analytical results.  

\subsubsection{The local loss function}

In the case of local loss function, the loss function takes the following formula,
\begin{equation}\label{Supp_Local}
\mathcal{L}(\vec{\theta}) = \frac{1}{2}\sum_i\left(\langle{\vec{x}}_i |U^\dag(\vec{\theta})Y_1U(\vec{\theta})|{\vec{x}}_i\rangle-y_i\right)^2.
\end{equation}
Here we chose the local observable $\hat{O}$ to be a local Pauli-Y operator acting on the first qubit site,  $\hat{O}\equiv Y_1$.

With the above global loss function, we then obtain the corresponding QNTK matrix element for arbitrary two input samples $\vec{x}$ and $\vec{x}'$
\begin{equation}
K_{\vec{x},\vec{x}'}=K(\vec{x},\vec{x}')=\sum_{\lambda}\frac{\partial \langle{\vec{x}} |U^\dag(\vec{\theta})Y_1U(\vec{\theta})|{\vec{x}}\rangle}{\partial \theta_{\lambda}} \frac{\partial \langle \vec{x}' |U^\dag(\vec{\theta})Y_1U(\vec{\theta})|{\vec{x}}'\rangle}{\partial \theta_{\lambda}}.
\end{equation}

Following the similar approach,  we numerically calculate the expectation value and variance of $K_{\vec{x},\vec{x}'}$ over the input data set $\mathcal{X}=\{\vec{x}\}$. We compute the expectation value and variance of QNTK in various cases involving different numbers of qubits and circuit depths. The qubit number ranges from four to eleven, and the number of block layers $d$ takes the cases of $\{5, 10, 20, 50, 100, 200\}$.  The full numerical results have been presented in the main text.  We find that both the expectation value and variance of QNTK  for local loss decay exponentially to zero with respect to the system size $n$, and the decay slopes are much smaller than those of global loss. These numerical results  match exactly with our analytical findings. Thus we conclude that judiciously choosing the local observables for loss functions would in some sense   mitigate the exponential  concentration problem of QNTK in practical usage.

\subsection{Local quantum encoding}
We also numerically investigate the  case of local quantum encoding, see Figs.~\ref{Supp_LE_mean_var_QNTK}-\ref{Supp_LE_Mean_var_scaling_d}.

\begin{figure}
\centering
\includegraphics[width=0.75\linewidth]{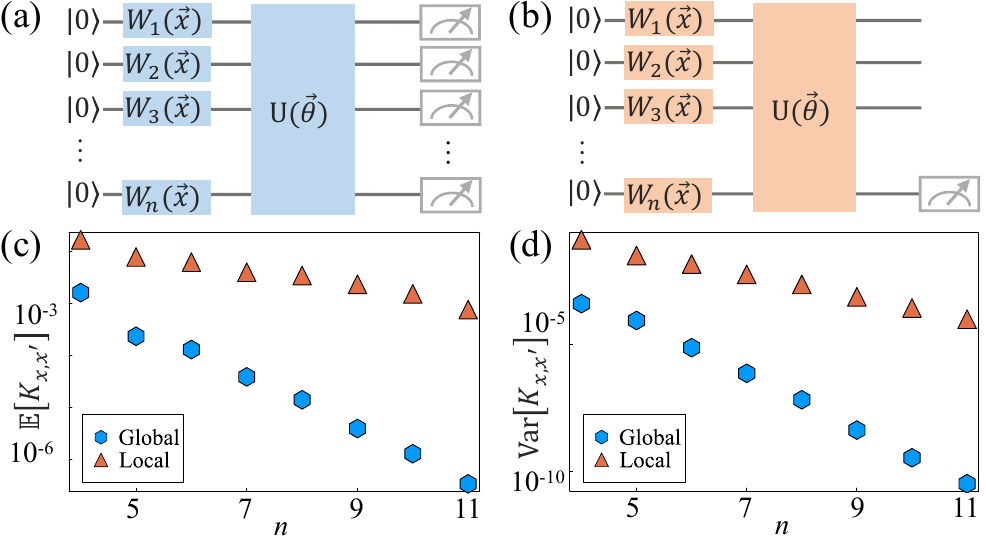} 
\caption{ {Numerical results on the mean and variance of QNTK values under local quantum encoding.}
(a) Quantum circuits of quantum learning models with global observables.
(b) Quantum circuits of quantum learning models with local observables.
 The learning models contain  the encoding layers $\otimes_{i=1}^n W_i(\vec{x})$ and the variational layers $U(\vec{\theta})$. (c) Scatter diagram of the mean  of QNTK values over the training input data set $\{\vec{x}\}$ versus the system size $n$, including both the global and local cases. (d) Variance of QNTK $K(\vec{x},\vec{x}')$ versus the system size $n$. In both cases, the system size $n$ varies from 4-qubit to 11-qubit, and the depth of variational circuit $U(\vec{\theta})$ is proportional to the system size $n$, see Fig.~\ref{Supp_Var_Circ}. To simulate the high expressibility of the  local encoding ensemble $\mathbb{W}$, one randomly samples the input data  from $(SU(2))^{\otimes n}$.}
\label{Supp_LE_mean_var_QNTK}
\end{figure}

\begin{figure}
\centering
\includegraphics[width=0.75\linewidth]{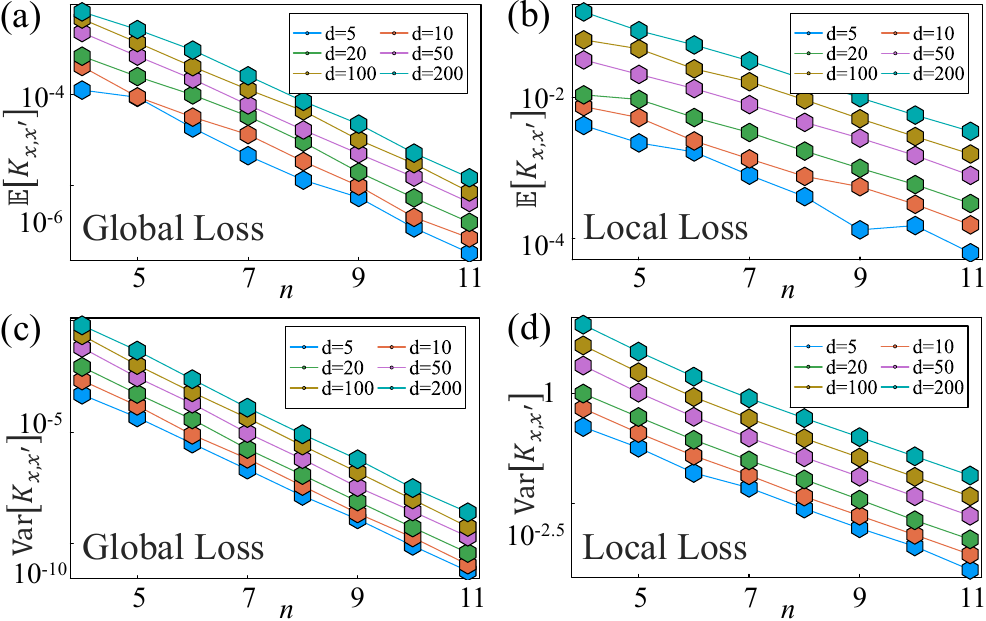}

\caption{Numerical results depict the mean and variance of QNTK values as a function of the number of variational circuit layers, with a focus on the local quantum encoding. 
(a) and (c) showcase the mean value and variance of $K(\vec{x},\vec{x}')$ concerning the global loss function, while (b) and (d) display the mean value and variance of $K(\vec{x},\vec{x}')$ associated with the local loss function. In both scenarios, the system size, denoted as $n$, ranges from 4-qubit to 11-qubit systems, while the number of variational layers, represented by $d$, spans from 5 to 200.}
\label{Supp_LE_Mean_var_scaling_d}
\end{figure}

\end{document}